%% file: ITCS2024.tex
\documentclass{article}
\usepackage[utf8]{inputenc}

\usepackage{hyperref}

\usepackage{amsmath,amssymb,amsthm,fullpage,mathrsfs,pgf,tikz,tikz-cd,circuitikz,mathtools,cleveref,bm,bbm}

\usepackage[ruled,vlined,linesnumbered,noresetcount]{algorithm2e}

\usepackage{nicematrix,tikz}
\usetikzlibrary{fit}

\usepackage[numbers]{natbib}
\usepackage[margin=1in]{geometry}
\usepackage{enumitem}

\newcommand{\LWS}{\mathsf{LWS}}
\newcommand{\kLWS}{\mathsf{kD \; LWS}}
\newcommand{\kminusoneLWS}{\mathsf{(k-1)D \; LWS}}
\newcommand{\twoLWS}{\mathsf{2D \; LWS}}

\newcommand{\minplus}{\mathsf{(min, +)MM}}

\newcommand{\APSP}{\mathsf{APSP}}
\newcommand{\OV}{\mathsf{OV}}
\newcommand{\threeOV}{\mathsf{3OV}}
\newcommand{\kOV}{\mathsf{kOV}}
\newcommand{\kminusoneOV}{\mathsf{(k-1)OV}}

\newcommand{\SETH}{\mathsf{SETH}}

\newcommand{\kSAT}{\mathsf{kSAT}}
\newcommand{\threeMinIP}{\mathsf{3Min\text{-}IP}}

\newcommand{\Static}{\mathsf{Static}}
\newcommand{\MinIP}{\mathsf{Min\text{-}IP}}

\newcommand{\kMinIP}{\mathsf{kMin\text{-}IP}}
\newcommand{\kminusoneMinIP}{\mathsf{(k-1)Min\text{-}IP}}

\newcommand{\SAT}{\mathsf{SAT}}
\newcommand{\NegativeTriangle}{\mathsf{NegativeTriangle}}

\newcommand{\CRR}{\mathsf{CRR}}
\newcommand{\micro}{\mathsf{micro}}

\newcommand{\PT}{\mathsf{PT}}

\newcommand{\lis}{\mathsf{LIS}}

\usepackage[parfill]{parskip}
\input{macro}

\title{Tensor Ranks and the Fine-Grained Complexity of Dynamic Programming}
\author{Josh Alman\thanks{Department of Computer Science,
Columbia University, NY 10027. Email: \texttt{josh@cs.columbia.edu}.} \and Ethan Turok \thanks{Department of Computer Science,
Columbia University, NY 10027. Email: \texttt{ezt2102@columbia.edu}.}\and Hantao Yu\thanks{Department of Computer Science,
Columbia University, NY 10027. Email: \texttt{hantao.yu@columbia.edu}.} \and Hengzhi Zhang \thanks{Department of Computer Science,
Columbia University, NY 10027. Email: \texttt{hz2663@columbia.edu}.}}

\begin{document}

\maketitle

\begin{abstract}
    Generalizing work of K\"unnemann, Paturi, and Schneider [ICALP 2017], we study a wide class of high-dimensional dynamic programming (DP) problems in which one must find the shortest path between two points in a high-dimensional grid given a tensor of transition costs between nodes in the grid. This captures many classical problems which are solved using DP such as the knapsack problem, the airplane refueling problem, and the minimal-weight polygon triangulation problem. We observe that for many of these problems, the tensor naturally has low tensor rank or low slice rank.

    We then give new algorithms and a web of fine-grained reductions to tightly determine the complexity of these problems. For instance, we show that a polynomial speedup over the DP algorithm is possible when the tensor rank is a constant or the slice rank is 1, but that such a speedup is impossible if the tensor rank is slightly super-constant (assuming SETH) or the slice rank is at least 3 (assuming the APSP conjecture). 
    We find that this characterizes the known complexities for many of these problems, and in some cases leads to new faster algorithms. 
\end{abstract}

\section{Introduction}

Dynamic programming (DP) is one of the most common algorithmic paradigms, used throughout the theory and practice of diverse computational domains. See \cite{CLRS01} chapter 14 for a detailed introduction.

When one solves a problem using DP, a natural question arises: is this the fastest algorithm to solve the problem? Recently, fine-grained complexity has been used to show that for many important problems, the answer is yes. For instance, researchers have established conditional lower bounds for the longest common subsequence \cite{ABW15,BK15}, edit distance \cite{BI15}, Fréchet distance \cite{Bringmann14}, and regular expression matching \cite{BI16}, showing that there is no algorithm (whether or not it uses DP) that is faster than the standard DP algorithm by a polynomial factor.

On the other hand, there are some notable examples where a natural DP formulation is \emph{not} the fastest known way to solve a problem. Consider, for instance, the 
polygon triangulation problem from computational geometry. In this problem, we are given as input a convex polygon with $n$ nodes, where each node $i$ has a weight $w_i$. For each triple $i,j,k$ of nodes, a triangle with those nodes as vertices has weight $w_i \cdot w_j \cdot w_k$. The weight of a triangulation of the polygon is the sum of the weights of its constituent triangles. The goal in the problem is to find the triangulation of the polygon with minimum weight. This problem has applications in point visibility \cite{Hershberger89}, mesh generation \cite{BE95}, computer graphics \cite{NM95}, and even in visual cryptography~\cite{SSMB12}.

Polygon triangulation has a natural DP formulation as follows. Let $T[i,j]$ denote the minimum weight of a triangulation of the polygon consisting of just nodes $i, i+1, i+2, \ldots, j$ with an edge drawn between nodes $i$ and $j$. Thus our goal is to compute $T[1,n]$, and these values satisfy the recurrence 
\[
T[i,j] = \min_{i < k < j}\Bigl\{T[i,k] + T[k,j] + w_i \cdot w_j \cdot w_k \Bigr\}.
\] (Since there is an edge from $i$ to $j$ in the polygon, there must be a triangle involving those two nodes and a third node $k$; the recurrence corresponds to iterating over the choices of that third node.) 
\begin{figure}[ht]
\begin{center}
\resizebox{5cm}{!}{

\tikzset{every picture/.style={line width=0.75pt}} 

\begin{tikzpicture}[x=0.75pt,y=0.75pt,yscale=-1,xscale=1]

\draw [line width=2.25]    (218.72,122.32) -- (149.09,172.9) ;
\draw [line width=2.25]    (218.72,122.32) -- (304.78,122.32) ;
\draw [line width=2.25]    (304.78,122.32) -- (374.41,172.9) ;
\draw [color={rgb, 255:red, 0; green, 0; blue, 0 }  ,draw opacity=1 ][line width=2.25]  [dash pattern={on 2.53pt off 3.02pt}]  (304.78,387.18) -- (374.41,336.6) ;
\draw [line width=2.25]    (122.5,254.75) -- (149.09,336.6) ;
\draw [color={rgb, 255:red, 0; green, 0; blue, 0 }  ,draw opacity=1 ][line width=2.25]  [dash pattern={on 2.53pt off 3.02pt}]  (149.09,336.6) -- (218.72,387.18) ;
\draw [line width=2.25]    (374.41,336.6) -- (401,254.75) ;
\draw [color={rgb, 255:red, 208; green, 2; blue, 27 }  ,draw opacity=1 ][line width=2.25]  [dash pattern={on 6.75pt off 4.5pt}]  (122.5,254.75) -- (401,254.75) ;
\draw [color={rgb, 255:red, 0; green, 0; blue, 0 }  ,draw opacity=1 ][line width=2.25]    (218.72,387.18) -- (304.78,387.18) ;
\draw [color={rgb, 255:red, 0; green, 0; blue, 0 }  ,draw opacity=1 ][line width=2.25]  [dash pattern={on 2.53pt off 3.02pt}]  (122.5,254.75) -- (149.09,172.9) ;
\draw [color={rgb, 255:red, 0; green, 0; blue, 0 }  ,draw opacity=1 ][line width=2.25]  [dash pattern={on 2.53pt off 3.02pt}]  (374.41,172.9) -- (401,254.75) ;
\draw [color={rgb, 255:red, 208; green, 2; blue, 27 }  ,draw opacity=1 ][line width=2.25]  [dash pattern={on 6.75pt off 4.5pt}]  (122.5,254.75) -- (218.72,387.18) ;
\draw [color={rgb, 255:red, 208; green, 2; blue, 27 }  ,draw opacity=1 ][line width=2.25]  [dash pattern={on 6.75pt off 4.5pt}]  (218.72,387.18) -- (401,254.75) ;

\draw (203,96.4) node [anchor=north west][inner sep=0.75pt]  [font=\Large]  {$1$};
\draw (126,153.4) node [anchor=north west][inner sep=0.75pt]  [font=\Large]  {$2$};
\draw (307,100.4) node [anchor=north west][inner sep=0.75pt]  [font=\Large]  {$n$};
\draw (101,242.4) node [anchor=north west][inner sep=0.75pt]  [font=\Large]  {$i$};
\draw (412,238.4) node [anchor=north west][inner sep=0.75pt]  [font=\Large]  {$j$};
\draw (378,146.4) node [anchor=north west][inner sep=0.75pt]  [font=\Large]  {$n-1$};
\draw (210,393.4) node [anchor=north west][inner sep=0.75pt]  [font=\Large]  {$k$};
\draw (100,335.4) node [anchor=north west][inner sep=0.75pt]  [font=\Large]  {$i+1$};
\draw (296,392.4) node [anchor=north west][inner sep=0.75pt]  [font=\Large]  {$k+1$};
\draw (383,324.4) node [anchor=north west][inner sep=0.75pt]  [font=\Large]  {$j-1$};

\end{tikzpicture}

}
\end{center}

\caption{An example polygon triangulation problem. The polygon $P(i,j)$ is partitioned into 3 parts by choosing $k$ and forming a triangle $(i,j,k)$ whose weight is $w_i\cdot w_j\cdot w_k$.}
\end{figure}
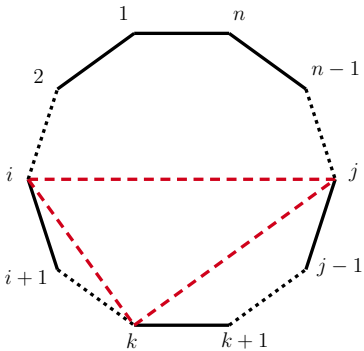

This recurrence leads to a DP algorithm which solves the problem in time $O(n^3)$. However, Hu and Shing showed in \cite{HS82, HS84} that the problem can actually be solved much faster, in time $O(n \log n)$. This is a surprising example where geometric techniques lead to a faster algorithm than the natural DP formulation.

\subsection{Least Weight Subsequence Problems}
Künnemann, Paturi, and Schneider~\cite{KPS17} initiated a general study of these phenomena. They studied a general class of problems intended to capture many one-dimensional DP problems called Least Weight Subsequence ($\LWS$) problems: Given as input a positive integer $n$ and an $n \times n$ matrix $w$, compute the value $T[n]$ defined by the recurrence:
\begin{align}
\label{eq:1dlsw}
    T[j]
    &=
    \begin{cases}
    0
    &
    \text{if } j = 0
    \\
    \displaystyle\min_{0 \leq i<j} \Bigl\{T[i]+w[i,j]\Bigr\}
    &
    \text{otherwise.}
    \end{cases}
\end{align}

$\LWS$ was first introduced by Hirschbert and Lamore \cite{HL87} to capture many known DP problems, including longest increasing subsequence \cite{Fredman75}, airplane refueling \cite{HL87}, coin change \cite{KPS17}, nested boxes \cite{KPS17}, pretty printing \cite{KP81}, and various knapsack problems. 

For illustration purposes, consider the longest increasing subsequence ($\lis$) problem: given an array of $n$ integers $X = [x_1, \ldots, x_n]$, return the length of the longest strictly increasing subsequence in $X$ \cite{Fredman75}. $\lis$ can be formulated as an $\LWS$ problem by setting 
\begin{align*}
    w[i, j]
    &=
    \begin{cases}
    -1
    &
    \text{if } x_i < x_j
    \\
    \infty
    &
    \text{otherwise.}
    \end{cases}
\end{align*}
Notice that $w[i,j]$ equals negative one when $x_i$ can be added to a subsequence which ends in $x_j$, thus increasing the length of a strictly increasing subsequence by $1$. Since $\lis$ is a maximization problem and $\LWS$ is a minimization problem, the weights are $-1$, not $1$, and the solution is given by $-T[n]$. Many algorithmic problems can be formulated as an $\LWS$ instance by appropriately setting the weight matrix $w$.

Straightforward DP solves the $\LWS$ problem in $O(n^2)$ time. Since the input matrix $w$ has $\Omega(n^2)$ entries, it requires quadratic time to read the input, so a faster algorithm isn't possible in general. However, if the $w$ matrix is given in a smaller, compressed form, then one may hope for subquadratic-time algorithms.\footnote{There has also been prior work on algorithms which assume $w$ has some additional structure which does not mean $w$ is compressible, but which lets one solve the problem without looking at most entries of $w$. For instance, \cite{HL87} and \cite{Wilber88} give $O(n\log n)$ and $O(n)$ time algorithms, respectively, for solving $\LWS$ with concave weights, i.e., when the entries of the matrix $w$ are promised to satisfy a quadrangle inequality. See also~\cite{GG89,KK90,Miller88}.}

One example that \cite{KPS17} focuses on is the case when $w$ is a low-rank matrix. If $w$ has rank $r < n^{o(1)}$, then one can be given as input matrices $A,B \in \R^{n \times r}$ such that $w = A \times B^T$, so the input size to the problem is only $n^{1 + o(1)}$. 

Interestingly, Künnemann showed via fine-grained reductions that this problem is subquadratic-time \emph{equivalent} to the well-studied $\MinIP$ problem for vectors of dimension $r$: Given as input $x_1, \ldots, x_n, y_1, \ldots, y_n \in \R^r$, find the $x_i, y_j$ minimizing the inner product $\langle x_i, y_j \rangle$. This problem can be solved in time $O(n^{2 - 1/r})$ using geometric techniques \cite{Yao82,Mat92,AESW90}, and thus has a truly-subquadratic time algorithm whenever $r$ is a constant. On the other hand, it is known that assuming the Strong Exponential Time Hypothesis ($\SETH$), the $\MinIP$ problem requires time $n^{2 - o(1)}$ even when $r$ is slightly super-constant $r = 2^{\log^* n}$ \cite{Chen18}, and thus the DP algorithm is essentially optimal. (Here $\log^*$ denotes the very slowly-growing iterated logarithm function.)

In this paper, we investigate the optimality of higher-dimensional DP formulations. 
We focus especially on generalizing $\LWS$ with low-rank matrices. As we will see, the known one-dimensional reductions do not generalize in a straightforward way, leading to a variety of generalizations and an intricate landscape of results. We will find that many classical higher-dimensional DP problems like the 
polygon triangulation problem are captured by our generalizations.

There are two choices to be made when generalizing $\LWS$ with low-rank matrices to higher dimensions: what is the higher-dimensional generalization of matrix rank (section \ref{sec: matrix rank generalization}) and what is the higher-dimensional generalization of the $\LWS$ recurrence (section \ref{sec: higher-dimensional LWS recurrences}).

\subsection{Generalizations of matrix rank}
\label{sec: matrix rank generalization}

The rank of a matrix has many equivalent definitions. However, when these definitions are generalized to higher-dimensional tensors, they lead to different notions. Prominent examples with applications in algebra, combinatorics, algorithm design, and complexity theory include the rank, subrank, border rank, slice rank, flattening rank, analytic rank, and geometric rank~(see, e.g.,~\cite{ottaviani2020tensor,kopparty2020geometric}). It is thus not clear, in general, which notion to use when generalizing results involving low-rank matrices.

We focus here on the two notions which arise naturally in known DP problems: tensor rank and slice rank.

\paragraph{Tensor Rank} A $d$-dimensional (order-$d$) tensor $w \in \R^{n_1 \times n_2 \times \cdots \times n_d}$ has rank $1$ if there are vectors $x_1 \in \R^{n_1}, \ldots, x_d \in \R^{n_d}$ such that, for all $i_1 \in [n_1], \ldots, i_d \in [n_d]$ we have $w[i_1, \ldots, i_d] = x_1[i_1] \cdot x_2 [i_2] \cdots x_d[i_d]$. More generally, the rank of tensor $w$ is the minimum non-negative integer $k$ such that there are rank $1$ tensors $w_1, \ldots, w_k$ for which $w = w_1 + \cdots + w_k$. This notion is sometimes also called canonical polyadic decomposition (CPD) rank.

For instance, in the polygon triangulation discussed earlier, the tensor $w$ whose entry $w[i,j,k]$ gives the weight of triangle $(i,j,k)$ has rank $1$ because the weight of the triangle $(i,j,k)$ is $w[i,j,k] = x_i\cdot x_j\cdot x_k$. 

For another example, consider the airplane refueling problem: an airplane is traveling on a grid with dimension $k$ such that each point in the grid is a refueling airport. The airplane starts at location $(1,\ldots,1)$ and wants to arrive at location $(n,\ldots,n)$. The cost of flying from $(i_1,\ldots,i_{\ell-1},j_{\ell},i_{\ell+1},\ldots,i_k)$ to $(i_1,\ldots,i_k)$ is $w[i_1,\ldots,i_k,j_{\ell}]$ (the airplane can only flies on the grid). The problem asks to minimize the cost of traveling.

One commonly studied cost of traveling from $(i_1,\ldots,i_{\ell-1},j_{\ell},i_{\ell+1},\ldots,i_k)$ to $(i_1,\ldots,i_k)$ is $w[i_1,\ldots,i_k,j_{\ell}] = (k-(i_{\ell}-j_{\ell}))^2$ for a fixed constant $k$~\cite{HL87}, which has rank $4$ since
\[
(k-(i_{\ell}-j_{\ell}))^2 = i_{\ell}^2\cdot 1+ 1\cdot j_{\ell}^2 + (i_{\ell}-k)\cdot (-2j_{\ell})+(i_{\ell}-\frac{k}{2})\cdot (-2k).
\]

\paragraph{Slice Rank} A $d$-dimensional (order-$d$) tensor $w \in \R^{n_1 \times n_2 \times \cdots \times n_d}$ has slice rank $1$ if there is a $j \in [d]$, a vector $a \in \R^{n_j}$, and a $(d-1)$-dimensional tensor $b \in \R^{n_1 \times \cdots \times n_{j-1} \times n_{j+1} \times \cdots \times n_d}$ such that, for all $i_1 \in [n_1], \ldots, i_d \in [n_d]$ we have $w[i_1, \ldots, i_d] = a[i_j] \cdot b[i_1, \ldots, i_{j-1}, i_{j+1}, \ldots, i_d]$. More generally, the slice rank of tensor $w$ is the minimum non-negative integer $k$ such that there are slice rank $1$ tensors $w_1, \ldots, w_k$ for which $w = w_1 + \cdots + w_k$. Slice rank was recently introduced in the context of solving the cap set problem from extremal combinatorics~\cite{croot2017progression,tao2016notes,tao2016notes2,blasiak2017cap}, but it has since found applications in algorithm design and complexity theory~\cite{blasiak2017cap,alman2021limits,alman2021limits2,blaser2020slice,blaser2021orbit} and other areas of combinatorics~\cite{ellenberg2017large,naslund2017upper,sawin2018bounds,lovasz2019lower}.

It is immediate that if a tensor $w$ has rank $d$, then it has slice rank at most $d$. However, there are many natural situations where the slice rank of a tensor may be much lower than its rank, and we could hope to take advantage of this to design faster algorithms.

For example, another reasonable cost function for the airplane refueling problem is the one that depends only on the destinations, e.g., each airport charges a fee for landing at that airport. In this scenario, the cost function would have slice rank $1$ but very large rank. We discuss the details in \Cref{sec: airplane refueling}.

\subsection{Higher-dimensional $\LWS$ recurrences}
\label{sec: higher-dimensional LWS recurrences}

Many problems solvable by $\LWS$ can be naturally generalized to higher dimensions, which motivates us to study high dimensional version of the $\LWS$ recurrence. We focus on two new recurrences which capture most examples. The first, which we call $\kLWS$, is perhaps the most general.

\begin{definition}[$\kLWS$]
For a positive integer $k$, consider $(k+1)$-dimensional tensors $w_1,\ldots,w_{k}$ of size $n \times n \times \cdots \times n$, where $w_{\ell}[i_1,\ldots,i_{k},j] \in \{-W,\ldots,W,\infty\}$ for all $1 \leq \ell \leq k$. The $\kLWS$ problem asks, given as input $w_1, \ldots, w_k$, to determine $T[n,\ldots,n]$ given the dynamic programming recurrence relation:
\begin{align*}
    T\Big[j_1,j_2,\ldots,j_k\Big] = 
    \begin{cases}
        0 \hspace{10em}\textup{ if } j_1 = j_2 = \ldots = j_k = 1 \\
        \displaystyle\min_{1\leq \ell \leq k}\Bigl\{\displaystyle\min_{1 \leq i_{\ell}<j_{\ell}}\Bigl\{T\Big[j_1,\ldots,j_{\ell-1},i_{\ell},j_{\ell+1},\ldots,j_k\Big]+w_{\ell}\Big[j_1,j_2,\ldots,j_k,i_{\ell}\Big]\Bigr\}\Bigr\} \textup{ otherwise.}
    \end{cases}
\end{align*}
\end{definition}

Intuitively, to compute $T[j_1,j_2,\ldots,j_k]$, we look at all \textit{previous} terms in the table $T$ that differ from $(j_1,j_2,\ldots,j_k)$ by \textit{one} coordinate. For example, when $k = 2$, $\twoLWS$ can be expressed as
\begin{align*}
    T[i, j]
    &=
    \begin{cases}
        0
        &
        \text{if } i = j = 1
        \\
        \min
        \begin{cases}
            \displaystyle\min_{1 \leq k < i} \{ T[k, j] + w_1[i, j, k] \}
            \\
            \displaystyle\min_{1 \leq k < j} \{ T[i, k] + w_2[i, j, k] \}
        \end{cases}
        &
        \text{otherwise.}
    \end{cases}
\end{align*}

$\kLWS$ captures high-dimensional analogs of many previous problems solved by $\LWS$ and also some new problems which we discuss below. This includes higher dimensional airplane refueling (see \Cref{sec: airplane refueling} below), $\kMinIP$ (\Cref{sec: Static kLWS Hierarchy}), all-pairs shortest paths (\Cref{sec: slice rank twoLWS}), multiple nested box chains (\Cref{sec: multiple nested boxes}), etc. An illustration of $\twoLWS$ is shown in \Cref{fig: kLWS}.

\begin{figure}[ht]
\begin{center}
\resizebox{5cm}{!}{
\tikzset{every picture/.style={line width=0.75pt}} 

\begin{tikzpicture}[x=0.75pt,y=0.75pt,yscale=-1,xscale=1]

\draw [line width=2.25]  (64.7,630) -- (532.7,630)(111.5,223.65) -- (111.5,675.15) (525.7,625) -- (532.7,630) -- (525.7,635) (106.5,230.65) -- (111.5,223.65) -- (116.5,230.65)  ;
\draw  [color={rgb, 255:red, 0; green, 0; blue, 0 }  ,draw opacity=1 ][fill={rgb, 255:red, 0; green, 0; blue, 0 }  ,fill opacity=1 ] (369,405) .. controls (369,396.16) and (376.16,389) .. (385,389) .. controls (393.84,389) and (401,396.16) .. (401,405) .. controls (401,413.84) and (393.84,421) .. (385,421) .. controls (376.16,421) and (369,413.84) .. (369,405) -- cycle ;
\draw  [draw opacity=0] (111.5,280) -- (462,280) -- (462,632) -- (111.5,632) -- cycle ; \draw   (111.5,280) -- (111.5,632)(161.5,280) -- (161.5,632)(211.5,280) -- (211.5,632)(261.5,280) -- (261.5,632)(311.5,280) -- (311.5,632)(361.5,280) -- (361.5,632)(411.5,280) -- (411.5,632)(461.5,280) -- (461.5,632) ; \draw   (111.5,280) -- (462,280)(111.5,330) -- (462,330)(111.5,380) -- (462,380)(111.5,430) -- (462,430)(111.5,480) -- (462,480)(111.5,530) -- (462,530)(111.5,580) -- (462,580)(111.5,630) -- (462,630) ; \draw    ;
\draw [color={rgb, 255:red, 0; green, 0; blue, 0 }  ,draw opacity=1 ][line width=1.5]    (158.5,406.25) -- (335,406) ;
\draw [shift={(338,406)}, rotate = 179.92] [color={rgb, 255:red, 0; green, 0; blue, 0 }  ,draw opacity=1 ][line width=1.5]    (14.21,-4.28) .. controls (9.04,-1.82) and (4.3,-0.39) .. (0,0) .. controls (4.3,0.39) and (9.04,1.82) .. (14.21,4.28)   ;
\draw [color={rgb, 255:red, 0; green, 0; blue, 0 }  ,draw opacity=1 ][line width=1.5]    (387,572) -- (387,458.5) ;
\draw [shift={(387,455.5)}, rotate = 90] [color={rgb, 255:red, 0; green, 0; blue, 0 }  ,draw opacity=1 ][line width=1.5]    (14.21,-4.28) .. controls (9.04,-1.82) and (4.3,-0.39) .. (0,0) .. controls (4.3,0.39) and (9.04,1.82) .. (14.21,4.28)   ;
\draw  [color={rgb, 255:red, 0; green, 0; blue, 0 }  ,draw opacity=1 ][fill={rgb, 255:red, 255; green, 255; blue, 255 }  ,fill opacity=1 ] (121,405) .. controls (121,396.16) and (128.16,389) .. (137,389) .. controls (145.84,389) and (153,396.16) .. (153,405) .. controls (153,413.84) and (145.84,421) .. (137,421) .. controls (128.16,421) and (121,413.84) .. (121,405) -- cycle ;
\draw  [color={rgb, 255:red, 0; green, 0; blue, 0 }  ,draw opacity=1 ][fill={rgb, 255:red, 255; green, 255; blue, 255 }  ,fill opacity=1 ] (373,603) .. controls (373,594.16) and (380.16,587) .. (389,587) .. controls (397.84,587) and (405,594.16) .. (405,603) .. controls (405,611.84) and (397.84,619) .. (389,619) .. controls (380.16,619) and (373,611.84) .. (373,603) -- cycle ;

\draw (380,638.4) node [anchor=north west][inner sep=0.75pt]  [font=\LARGE]  {$i$};
\draw (82,387.4) node [anchor=north west][inner sep=0.75pt]  [font=\LARGE]  {$j$};
\draw (123,637.4) node [anchor=north west][inner sep=0.75pt]  [font=\Large]  {$1$};
\draw (83,590.4) node [anchor=north west][inner sep=0.75pt]  [font=\LARGE]  {$1$};
\draw (72,643.4) node [anchor=north west][inner sep=0.75pt]  [font=\LARGE]  {$T$};

\end{tikzpicture}

}   
\end{center}

\caption{$\twoLWS$. To compute $T[i,j]$, we take the minimum of all possible white circles (plus their respective tensor values $w$) such that their coordinates differ from the target by one coordinate.}
\label{fig: kLWS}
\end{figure}
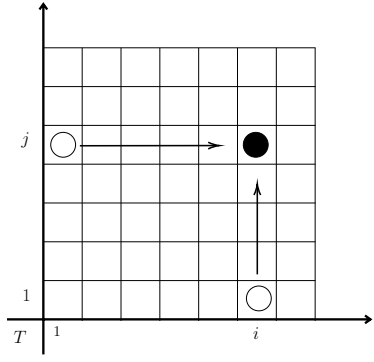

\paragraph{Static $\kLWS$}
Similar to \cite{KPS17}, we also define a notion of ``static" $\kLWS$ in which we are given some entries in the DP table, and we would like to compute new entries which depend only on the given entries. The main idea of $[\Static]\kLWS$ is that we have the information $(T[i_1,\ldots,i_k])$ for all $(i_1,\ldots,i_k)$ on a ``band" $D_{a,a+N}$ and we want to compute $T[i_1,\ldots,i_k]$ for all $(i_1,\ldots,i_k)$ on the next ``band" $D_{a+N,a+2N}$. A band $D_{\alpha,\beta}$ is defined to be all $(i_1,\ldots,i_k)$ such that their sum $i_1 + \cdots + i_k$ is in the interval $[\alpha,\beta)$.

\begin{definition}
$([\Static]\kLWS)$ Given intervals $D_{a,a+N},D_{a+N,a+2N}$ together with correctly computed values $T[i_1,\ldots,i_k]$ for all $1 \leq \ell \leq k$ and $(i_1,\ldots,i_k) \in D_{a,a+N}$, $[\Static]\kLWS$ asks to determine 
\begin{align*}
    T'\Big[j_1,\ldots,j_k\Big] = \min_{1\leq \ell \leq k}\Bigg\{&\displaystyle\min_{a-I_{\ell}\leq i_{\ell}<a+N-I_{\ell}}\Bigl\{T\Big[j_1,\ldots,j_{\ell-1},i_{\ell},j_{\ell+1},\ldots,j_k\Big]+w_{\ell}\Big[j_1,j_2,\ldots,j_k,i_{\ell}\Big]\Bigr\}\Bigg\}  
\end{align*} for all $(j_1,j_2,\ldots,j_k) \in D_{a+N,a+2N}$.
\end{definition}

For illustration purposes, consider the $[\Static] \twoLWS$ problem: given correctly computed values $T[i, j]$ for all $(i, j) \in D_{a, a+N}$, determine
\begin{align*}
    T'[i, j]
    &=
    \min
    \begin{cases}
        \displaystyle\min_{a - i \leq k < a + N - i} \{T[k, j] + w_1[i, j, k]\}
        \\
        \displaystyle\min_{a - j \leq k < a + N - j} \{T[i, k] + w_2[i, j, k]\}
    \end{cases}
\end{align*}
for all $(i, j) \in D_{a+N, a+2N}$. Figure \ref{fig: static kLWS} depicts the idea.

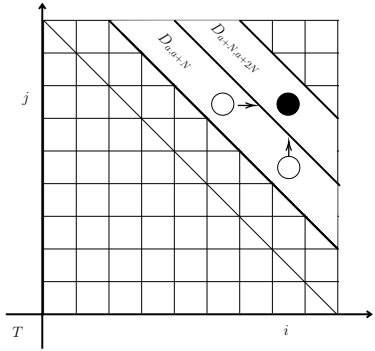
\begin{figure}[ht]
\begin{center}
\resizebox{5cm}{!}{

\tikzset{every picture/.style={line width=0.75pt}} 

\begin{tikzpicture}[x=0.75pt,y=0.75pt,yscale=-1,xscale=1]

\draw [line width=2.25]  (37.05,888.72) -- (596.05,888.72)(92.95,414.24) -- (92.95,941.44) (589.05,883.72) -- (596.05,888.72) -- (589.05,893.72) (87.95,421.24) -- (92.95,414.24) -- (97.95,421.24)  ;
\draw  [draw opacity=0][fill={rgb, 255:red, 255; green, 255; blue, 255 }  ,fill opacity=1 ] (94.95,438.72) -- (547,438.72) -- (547,887) -- (94.95,887) -- cycle ; \draw  [color={rgb, 255:red, 0; green, 0; blue, 0 }  ,draw opacity=1 ] (94.95,438.72) -- (94.95,887)(144.95,438.72) -- (144.95,887)(194.95,438.72) -- (194.95,887)(244.95,438.72) -- (244.95,887)(294.95,438.72) -- (294.95,887)(344.95,438.72) -- (344.95,887)(394.95,438.72) -- (394.95,887)(444.95,438.72) -- (444.95,887)(494.95,438.72) -- (494.95,887)(544.95,438.72) -- (544.95,887) ; \draw  [color={rgb, 255:red, 0; green, 0; blue, 0 }  ,draw opacity=1 ] (94.95,438.72) -- (547,438.72)(94.95,488.72) -- (547,488.72)(94.95,538.72) -- (547,538.72)(94.95,588.72) -- (547,588.72)(94.95,638.72) -- (547,638.72)(94.95,688.72) -- (547,688.72)(94.95,738.72) -- (547,738.72)(94.95,788.72) -- (547,788.72)(94.95,838.72) -- (547,838.72) ; \draw  [color={rgb, 255:red, 0; green, 0; blue, 0 }  ,draw opacity=1 ]  ;
\draw    (94.95,438.72) -- (544,889) ;
\draw  [draw opacity=0][fill={rgb, 255:red, 255; green, 255; blue, 255 }  ,fill opacity=1 ] (198.4,438.94) -- (294.95,438.72) -- (544.45,688.22) -- (544.23,784.77) -- cycle ;
\draw [line width=2.25]    (194.95,438.72) -- (544.62,788.38) ;
\draw  [draw opacity=0][fill={rgb, 255:red, 255; green, 255; blue, 255 }  ,fill opacity=1 ] (294.95,438.72) -- (393.73,438.73) -- (543.6,588.6) -- (543.62,687.38) -- cycle ;
\draw [line width=2.25]    (214.95,458.72) -- (544.62,788.38) ;
\draw [line width=2.25]    (294.95,438.72) -- (548.29,692.05) ;
\draw [line width=2.25]    (394.95,438.72) -- (544.95,588.72) ;
\draw  [fill={rgb, 255:red, 0; green, 0; blue, 0 }  ,fill opacity=1 ] (452,567) .. controls (452,557.61) and (459.61,550) .. (469,550) .. controls (478.39,550) and (486,557.61) .. (486,567) .. controls (486,576.39) and (478.39,584) .. (469,584) .. controls (459.61,584) and (452,576.39) .. (452,567) -- cycle ;
\draw  [fill={rgb, 255:red, 255; green, 255; blue, 255 }  ,fill opacity=1 ] (352,567) .. controls (352,557.61) and (359.61,550) .. (369,550) .. controls (378.39,550) and (386,557.61) .. (386,567) .. controls (386,576.39) and (378.39,584) .. (369,584) .. controls (359.61,584) and (352,576.39) .. (352,567) -- cycle ;
\draw  [fill={rgb, 255:red, 255; green, 255; blue, 255 }  ,fill opacity=1 ] (453,664) .. controls (453,654.61) and (460.61,647) .. (470,647) .. controls (479.39,647) and (487,654.61) .. (487,664) .. controls (487,673.39) and (479.39,681) .. (470,681) .. controls (460.61,681) and (453,673.39) .. (453,664) -- cycle ;
\draw [line width=1.5]    (392,569) -- (415,569) ;
\draw [shift={(418,569)}, rotate = 180] [color={rgb, 255:red, 0; green, 0; blue, 0 }  ][line width=1.5]    (14.21,-4.28) .. controls (9.04,-1.82) and (4.3,-0.39) .. (0,0) .. controls (4.3,0.39) and (9.04,1.82) .. (14.21,4.28)   ;
\draw [line width=1.5]    (470,647) -- (470,627) ;
\draw [shift={(470,624)}, rotate = 90] [color={rgb, 255:red, 0; green, 0; blue, 0 }  ][line width=1.5]    (14.21,-4.28) .. controls (9.04,-1.82) and (4.3,-0.39) .. (0,0) .. controls (4.3,0.39) and (9.04,1.82) .. (14.21,4.28)   ;

\draw (46,906.4) node [anchor=north west][inner sep=0.75pt]  [font=\LARGE]  {$T$};
\draw (280.43,452.7) node [anchor=north west][inner sep=0.75pt]  [font=\LARGE,rotate=-45]  {$D_{a,a+N}$};
\draw (360.96,436.54) node [anchor=north west][inner sep=0.75pt]  [font=\LARGE,rotate=-45]  {$D_{a+N,a+2N}$};
\draw (461,904.4) node [anchor=north west][inner sep=0.75pt]  [font=\LARGE]  {$i$};
\draw (61,546.4) node [anchor=north west][inner sep=0.75pt]  [font=\LARGE]  {$j$};

\end{tikzpicture}

}
\end{center}

\caption{$[\Static] \twoLWS$. To calculate $T'[i,j]$ (black circle), we take the minimum over all possible white circles (plus their respective weight values $w$) such that they share all but one coordinate with $T[i,j]$.}
\label{fig: static kLWS}
\end{figure}

\cite{KPS17} showed that in the $k=1$ case, $[\Static]\LWS$ is subquadratic-\emph{equivalent} to the original $\LWS$ problem. We will find that the relationships among the higher-dimensional versions are more complicated.

\subsection{Polygon Triangulation}

$\kLWS$ as we defined above captures many different high-dimensional DP problems, but it is not the only conceivable way to generalize $\LWS$. We consider here another example we call $\twoLWS^{\PT}$, which captures optimization over sets that are counted by the Catalan numbers.

Recall that the Catalan numbers $C_0, C_1, C_2, \ldots$ can be defined recursively by $C_0 = 1$ and \begin{align}\label{catrec}C_n = \sum_{k=1}^n C_{k-1} \cdot C_{n-k}.\end{align} $C_n$ counts many different combinatorial objects, such as the number of triangulations of a convex polygon with $n+2$ vertices, and the number of binary trees with $n+1$ leaves. (See, e.g.,~\cite{stanley2015catalan}.) The variable $k$ being summed over in \Cref{catrec} typically corresponds to ways to partition the combinatorial object into two smaller parts. This leads to our new definition, in which we want to instead \emph{minimize} over all ways to partition the object:

\begin{definition}[$\twoLWS^{\PT}$]
Given as input an $n \times n \times n$ tensor $w$, the $\twoLWS^{\PT}$ problem asks to compute the value of $T[n,n]$ given the dynamic programming recurrence relation:
\begin{align*}
    T[i,j] = 
    \begin{cases}
        0 & \textup{ if } j-i \leq 1 \\
        \displaystyle\min_{i<k<j}\Big\{T[i,k]+T[k,j]+w[i,j,k]\Big\}&\textup{otherwise.}
    \end{cases}
\end{align*}
\end{definition}

For instance, this captures the polygon triangulation problem defined above when $w[i,j,k] = w_i \cdot w_j \cdot w_k$, which is unsurprising as polygon triangulations are counted by the Catalan numbers; this inspires the name $\twoLWS^{\PT}$. We show in \Cref{sec:PT} below that this recurrence also captures other natural problems such as optimal binary search tree construction (optimizing over binary trees, which are counted by Catalan numbers) and matrix chain multiplication (optimizing over sequences of properly matched parentheses, again counted by Catalan numbers). Furthermore, in each of these examples, the rank (for polygon triangulation) or slice rank (for the other two examples) of $w$ is $1$.

\subsection{Main Results and Proof Structure Overview}

\paragraph{Reduction notation} Before stating our results, we introduce one piece of notation for denoting the results of fine-grained reductions between problems with different running times. For computational problems $P,Q$, we say that $P$ ``reduces" to $Q$, denoted by $P\rightarrow Q$ if a polynomial speedup for $Q$ yields a polynomial speedup for $P$. More precisely, suppose $P,Q$ are solved in time $T_p, T_q$, respectively, via the straightforward algorithms. We say that $P$ ``reduces" to $Q$, denoted by $P\rightarrow Q$ if for every 
$\varepsilon>0$ there exists a $\delta>0$ such that, given a $O(T_q^{1-\varepsilon})$ time algorithm for $Q$, one gets a $O(T_p^{1-\delta})$ time algorithm for $P$.

For example, $\SAT \rightarrow \MinIP_{n,c\log n}$ means that if there is an algorithm for $\MinIP_{n,c\log n}$ with running time $O(n^{2-\varepsilon})$ for some $\varepsilon>0$, then there is an algorithm for $\SAT$ with running time $O(2^{(1-\delta)n})$ for some $\delta>0$. $\threeMinIP \rightarrow \MinIP$ means that if there is an algorithm for $\MinIP$ with running time $O(n^{2-\varepsilon})$ for some $\varepsilon>0$, then there is an algorithm for $\threeMinIP$ with running time $O(n^{3-\delta})$ for some $\delta>0$. When it may be unclear, the formal statements of our results are stated in the theorems below.

\subsubsection{$\kLWS$ Hierarchy and Hardness.} We establish a hierarchy of $\kLWS$ problems and describe their connections to $\kMinIP$, summarized by the following diagram.
\begin{center}
\includegraphics[scale=0.15]{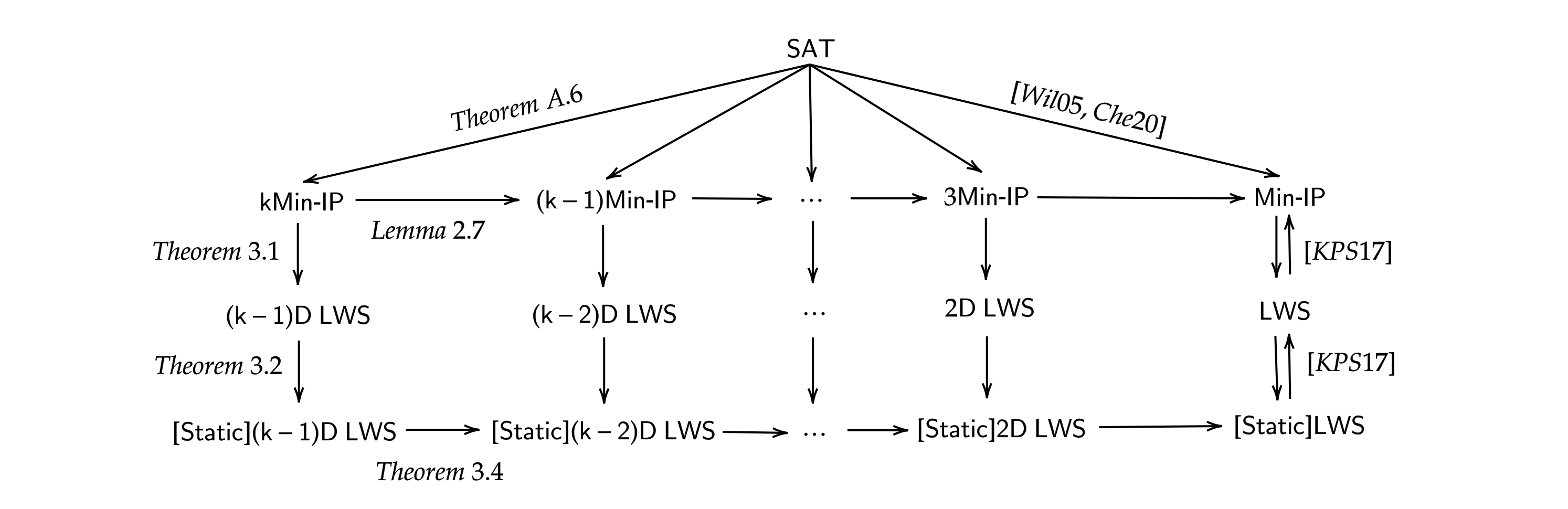}    
\end{center} 

In particular, assuming $\SETH$, it follows that the straightforward algorithms for all the problems in the diagram cannot be substantially improved. The results depicted in this diagram, more precisely stated, are described by the following four theorems.

Building on Chen's reduction~\cite{Chen18} from $\SAT$ to $\MinIP_{n,2^{O(\log^{*}n)}}$, we show that more generally, $\SAT$ also reduces to $\kMinIP$. (Note that $\kMinIP$ reduces to $\kminusoneMinIP$ in a straightforward way, but a reduction in the other direction is not known.)
\begin{theorem*}[\Cref{thm: SAT reduces to kMinIP}]
Assuming $\SETH$, there is no algorithm for $\kMinIP_{n,2^{O(\log^{*}n)}}$ running in time $O(n^{k-\varepsilon})$ for any $\varepsilon>0$.
\end{theorem*}

Just as $\MinIP$ reduces to $\LWS$, we show that $\kMinIP$ reduces to $\kminusoneLWS$.
\begin{theorem*}[\Cref{thm: kMinIP to k-1 D LWS}]
Suppose there exists an algorithm for $\kminusoneLWS$ with rank $d$ with running time $O(n^{k-\varepsilon})$ for some $\varepsilon>0$, then there exists an algorithm for $\kMinIP$ with rank $d$ with running time $O(n^{k-\delta})$ for some $\delta>0$. 
\end{theorem*}

By a divide and conquer method similar to  \cite[Lemma 3.5]{KPS17}, we show that $\kLWS$ can be reduced to $[\Static]\kLWS$.

\begin{theorem*}[\Cref{thm: kLWS to Static kLWS}]
Suppose there exists an algorithm for $[\Static]\kLWS_{n,N,d}$ with running time $O(N^{2-\varepsilon}\cdot n^{k-1})$ for some $\varepsilon>0$, then there exists an algorithm for $\kLWS_{n,d}$ with running time $O(n^{k+1-\delta})$ for some $\delta>0$.    
\end{theorem*}

In addition, we show that $[\Static]\kLWS$ also exhibits a hierarchy similar to $\kMinIP$.
\begin{theorem*}[\Cref{thm: static kLWS to static k-1 LWS}]
Suppose there exists an algorithm for $[\Static]\kminusoneLWS_{n,N,d}$ with running time $O(N^{2-\varepsilon}\cdot n^{k-2})$ for some $\varepsilon>0$, then there exists an algorithm for $[\Static]\kLWS_{n,N,d}$ with running time $O(N^{2-\delta}\cdot n^{k-1})$ for some $\delta>0$.    
\end{theorem*}

As one consequence of these reductions, $[\Static]\kLWS,\kLWS$ all reduce to $\MinIP$, which has a truly sub-quadratic algorithm when the vector length is constant, leading to a polynomial speedup:

\begin{corollary*} 
For every constant $c>0$ there is an $\varepsilon>0$ such that $\kLWS, [\Static]\kLWS$ can be solved in time $O(n^{k+1-\varepsilon})$, $O(N^{2-\varepsilon}\cdot n^{k-1})$ respectively if the rank of the tensor $w$ is at most $c$.
\end{corollary*}

For one application of this corollary, we show in \Cref{sec: airplane refueling} that the generalized airplane refueling problem \cite{HL87} in higher dimensions can be solved polynomially faster than the straightforward DP formulation.

Our reductions from $\SAT$ also give hardness of $\kLWS$ assuming $\SETH$, showing that the fast algorithm cannot extend much beyond constant rank:

\begin{corollary*}
Under $\SETH$, for any $k>1$ and $\varepsilon>0$, there is no algorithm running in time $O(n^{k+1-\varepsilon})$ for $\kLWS$ when the weight tensor has rank $2^{O(\log^{*}n)}$.
\end{corollary*}

Since for rank $r$, the input size is $nr$, one could have imagined a better running time than $O(n^{k+1})$ for any $r \ll n^k$; our lower bound shows that it is (conditionally) impossible even for the slightly super-constant $r = 2^{\Omega(\log^{*}n)}$.

\subsubsection{Slice Rank in $\kLWS$.} Slice rank is a another way to define tensor rank. If a tensor with rank $d$, then it trivially has slice rank at most $d$, which makes $\twoLWS$ with slice rank more powerful. Indeed, we show that $\twoLWS$ with slice rank becomes hard very quickly. Interestingly, this new hardness result builds on the $\APSP$ conjecture, rather than $\SETH$; the $\APSP$ conjecture has not previously been connected to $\LWS$ problems to our knowledge.

\begin{theorem*}[\Cref{thm: APSP to twoLWS}]
Assuming the $\APSP$ conjecture, there is no truly sub-cubic algorithm for $\twoLWS$ or $[\Static]\twoLWS$ with slice rank $3$.
\end{theorem*}

However, we can design a truly sub-cubic algorithm for $\twoLWS$ with slice rank $1$.

\begin{theorem*}[\Cref{thm: twoLWS slice rank 1 is truly sub-cubic},\Cref{thm: static twoLWS with slice rank 1 is truly sub-cubic}]
There are truly sub-cubic time algorithms for $\twoLWS$ and $[\Static]\twoLWS$ with slice rank 1.
\end{theorem*}

We then use this to design faster algorithms for natural dynamic programming problems. For instance, we generalize the nested boxes problem defined in \cite{KPS17} to a multiple nested boxes problem and show that it can be formulated as $\kLWS$ with slice rank $1$, and thus it can be solved in time $O(n^{k+1-\varepsilon})$ for some $\varepsilon>0$.
\begin{theorem*}[\Cref{thm: multiple nested boxes}]
Multiple nested boxes with dimension $k$ can be solved in time $O(n^{k+1-\varepsilon})$ for some $\varepsilon>0$.
\end{theorem*}

We also show how to give a polynomial speedup for the airplane refueling problem in dimension $k$ if the cost only depends on where the airplane lands, since that would mean the tensor has slice rank $1$. We discuss the details in \Cref{sec: airplane refueling}.

\subsubsection{Hardness of Polygon Triangulation.} 

We show similar algorithms and hardness for $\twoLWS^{\PT}$.

\begin{corollary*}[\Cref{cor: PT hardness rank}]
Under $\SETH$, there is no truly sub-cubic algorithm for $\twoLWS^{\PT}$ with weight function whose rank is $2^{O(\log^{*}n)}$ or above.
\end{corollary*}

\begin{corollary*}[\Cref{cor: PT hardness slice rank}]
Under $\APSP$ conjecture, there is no truly sub-cubic algorithm for $\twoLWS^{\PT}$ with weight function whose slice rank is $3$ or above.    
\end{corollary*}

These results are proved by making use of a reduction to $\twoLWS^{\PT}$ from a \emph{special case} of $\twoLWS$ where the two tensors $w_1,w_2$ must be equal.  We then show that our previous hardness results hold even in this special case, yielding the above corollaries.

In fact, previous work shows that in some special cases when $w$ has rank 1, $\twoLWS^{\PT}$ can be solved in truly sub-cubic time.

\begin{theorem*}[\cite{HS82,HS84,HS02,LG21}]
Suppose there exists $x_i \in \mathbb{N}, 1\leq i \leq n$ such that $w[i,j,k] = x_i\cdot x_j \cdot x_k$ for all $1\leq i,j,k \leq n$, then $\twoLWS^{\PT}$ with tensor $w$ can be solved in $O(n\log n)$ time. 
\end{theorem*}

Our results help explain why the examples of $\twoLWS^{\PT}$ problems where faster algorithms are known correspond to tensors $w$ of rank or slice rank $1$; see~\Cref{sec:PT} for more details.

\subsection{Organization}

Section \ref{preliminaries} contains the preliminaries: our notation, background on fine-grained complexity, and definitions of relevant problems. Section \ref{sec: kLWS} discusses the $\kLWS$ hierarchy and hardness, proving that a polynomial speedup over the standard DP algorithm is possible when the tensor rank is $O(1)$ but impossible when the tensor rank is $2^{O(\log^* n)}$ (assuming $\SETH$). Section \ref{sec: polygon triangulation} discusses the polygon triangulation problem $\twoLWS^{\PT}$ and its connections with $\twoLWS$. Sections \ref{sec: applications} and \ref{sec:PT} respectively discuss applications of $\kLWS$ and $\twoLWS^\PT$ to various real-world problems.

\section{Preliminaries}
\label{preliminaries}

In this section, we state our core problems and relevant problems in fine-grained complexity. We also state our notations for convenience.


For a problem $P$, we say that it is truly sub-quadratic (or sub-cubic) if there exists an algorithm for it with running time $O(n^{2-\varepsilon})$ (or $O(n^{3-\varepsilon})$) for any $\varepsilon>0$. We say that $P$ and $Q$ are sub-quadratic (sub-cubic) equivalent if $P$ is truly sub-quadratic (sub-cubic) if and only if $Q$ is truly sub-quadratic (sub-cubic).

For a positive integer $n$, we let $[n] = \{1,\ldots,n\}$. We assume the word-RAM model with word size $\Theta(\log n)$ throughout this paper, and assume that all integer inputs are chosen from $\{-W,\ldots,W, \infty\}$ where $W$ fits in a constant number of words. We usually use $d$ to denote the length of vectors (rank of a tensor) in our problems, $k$ to denote the dimension of our problems, and $n$ to denote the input size. We put these parameters as subscripts of problems. Since we will discuss both rank and slice rank, we make it clear which we are talking about in the discussions. In this paper, ``$\kLWS$ has (slice) rank $d$" means the array $w$ has (slice) rank $d$. 

For a $k$-dimensional dynamic programming array $T$ with entries $T[i_1,i_2,\ldots,i_k]$, we let $I_{\ell}$ denote the sum of all $i_k$'s except $i_{\ell}$. Let $D_{a,b}$ denote the set of all $(i_1,\ldots,i_k)$ such that $a \leq i_1+\ldots+i_k<b$.

\subsection{Strong Exponential Time Hypothesis and $\MinIP$}

We state some important problems and definitions in fine-grained complexity that we will refer to later. 

\begin{conjecture}[Strong Exponential Time Hypothesis $(\SETH)$]
For every $\varepsilon>0$, there exists a positive integer $k$ such that $\kSAT$ requires time $\Omega(2^{(1-\varepsilon)n})$.    
\end{conjecture}

$\SETH$ is well-known for being a stronger version of $\mathsf{P}\neq \mathsf{NP}$, i.e. it implies $\mathsf{P}\neq \mathsf{NP}$.

\begin{definition}[$\OV_{n,d}$]
Given two sets of vectors $A = \{a_1,\ldots,a_n\}, B = \{b_1,\ldots,b_n\}$ such that $a_i,b_j \in \{0,1\}^d$ for all $i,j$, the \textup{Orthogonal Vectors} problem $(\OV_{n,d})$ asks to determine whether there exists $1 \leq i,j \leq n$ such that $\langle a_i,b_j \rangle = 0$.    
\end{definition}

In \cite{Williams05}, it is shown that assuming $\SETH$, for any positive $c$ there exists $\varepsilon>0$ such that $\OV_{n,c\log n}$ cannot be solved in time $O(n^{2-\varepsilon})$.

\begin{definition}[$\MinIP_{n,d}$]
Given two sets of vectors $A = \{a_1,\ldots,a_n\},B = \{b_1,\ldots,b_n\}$ such that $a_i,b_j \in \{-W,\ldots,W, \infty\}^d$ for all $i,j$ and a natural number $r \in \mathbb{Z}$, the \textup{Minimal Inner Product} $(\MinIP_{n,d})$ asks to determine whether there exists $1 \leq i,j \leq n$ such that $\langle a_i,b_j \rangle \leq r$. 
\end{definition}

$\OV_{n,d}$ trivially reduces to $\MinIP_{n,d}$, and in fact these two problems are sub-quadratic equivalent \cite{CW19}. It is known that this decision version of $\MinIP_{n,d}$ is sub-quadratic equivalent to its counting version where it asks to output $\displaystyle\min_{1 \leq i,j \leq n}\langle a_i,b_j \rangle$. We will use the counting version throughout the rest of the paper.

\subsection{Higher Dimension $\OV,\MinIP$}
$\OV,\MinIP$ can be naturally generalize to higher dimensions as follows.

\begin{definition}[$\kOV_{n,d}$]
Given $k$ sets of vectors 
\[
X_1 = \{x_{11},\ldots,x_{1n}\},\ldots,X_k = \{x_{k1},\ldots,x_{kn}\}
\] such that $x_{ij} \in \{0,1\}^d$ for all $i,j$, $\kOV_{n,d}$ asks to determine whether there exists $1 \leq i_1,\ldots,i_k \leq n$ such that 
\[
\langle x_{1,i_{1}},x_{2,i_{2}},\ldots,x_{k,i_{k}} \rangle = 0.
\]
\end{definition}

\begin{definition}[$\kMinIP_{n,d}$]
Given $k$ sets of vectors     
\[
X_1 = \{x_{11},\ldots,x_{1n}\},\ldots,X_k = \{x_{k1},\ldots,x_{kn}\}
\] such that $x_{ij} \in \{-W,\ldots,W, \infty\}^{d}$ for all $i,j$ and a natural number $r \in \mathbb{Z}$, $\kMinIP_{n,d}$ asks to determine whether there exists $1 \leq i_1,\ldots,i_k \leq n$ such that 
\[
\langle x_{1,i_{1}},x_{2,i_{2}},\ldots,x_{k,i_{k}} \rangle \leq r.
\]
\end{definition}

Just from the definitions, $\kOV$ trivially reduces to $\kMinIP$. In addition, it is not hard to show $\kOV_{n,d} \rightarrow \kminusoneOV_{n,d}$ and $\kMinIP_{n,d} \rightarrow \kminusoneMinIP_{n,d}$ for all $k \geq 3$.

\begin{lemma}
\label{lem: kOV reduces to k-1 OV}
Suppose there exists an algorithm for $\kminusoneOV_{n,d}$ that runs in time $O(n^{k-1-\varepsilon})$ for some $\varepsilon>0$, then there exists an algorithm for $\kOV_{n,d}$ that runs in time $O(n^{k-\delta})$ for some $\delta>0$.   
\end{lemma}
\begin{proof}
Given an $\kOV_{n,d}$ instance with sets $X_1,\ldots,X_k$, we trivially compute $\langle x_{1,i_1},x_{2,i_{2}}\rangle$ for all $1 \leq i_1,i_2 \leq n$ using time $O(n^2d)$. Now for each $1 \leq i_1 \leq n$, run the algorithm for $\kminusoneOV_{n,d}$ with $x_{1,i_1}\cdot X_2, X_3,\ldots,X_k$. If there are no zeros then output no; otherwise output yes.

This algorithm is correct because we have covered all possible $\langle x_{1,i_{1}},x_{2,i_{2}},\ldots,x_{k,i_{k}}\rangle$. The running time is $O(n^2d)+O(n^{k-1-\varepsilon})\cdot n = O(n^{k-1-\varepsilon})$.
\end{proof}

\begin{lemma}
Suppose there exists an algorithm for $\kminusoneMinIP_{n,d}$ that runs in time $O(n^{k-1-\varepsilon})$ for some $\varepsilon>0$, then there exists an algorithm for $\kMinIP_{n,d}$ that runs in time $O(n^{k-\delta})$ for some $\delta>0$.
\end{lemma}
\begin{proof}
The proof is exactly the same as the proof of \Cref{lem: kOV reduces to k-1 OV}.
\end{proof}

\subsection{All Pair Shortest Path}

All Pair Shortest Path $(\APSP)$ is a well-known problem in fine grained complexity and graph theory. 

\begin{definition}[$\APSP$] Given a directed graph $G$ with nodes $V = \{v_1,\ldots,v_n\}$, the \textup{All Pair Shortest Path} $(\APSP)$ asks to determine the distance between $v_i$ and $v_j$ for all $1 \leq i<j \leq n$.    
\end{definition}

Currently there is no algorithm for $\APSP$ that runs in $O(n^{3-\varepsilon})$ time for any $\varepsilon>0$, and it is conjectured that no such algorithm exists. 

\begin{conjecture}[$\APSP$ Conjecture]
There is no truly sub-cubic algorithm for $\APSP$.
\end{conjecture}

It is known that $\APSP$ is sub-cubic equivalent to many problems such that min-plus matrix multiplication, negative triangle etc \cite{WW18}. Therefore, assuming $\APSP$ conjecture, none of these problems are truly sub-cubic. We will use the hardness of these problems to obtain our hardness results. 

\begin{definition}[$\minplus$]
Given two $n \times n$ matrices $A,B$, compute its min-plus product $C$ where
\[
C[i,j] = \min_{1 \leq k \leq n}\{A[i,k]+B[k,j]\}
\] for all $1 \leq i,j \leq n$.
\end{definition}

\begin{definition}[$\NegativeTriangle$]
Given an undirected, weighted graph, determines whether there exists a triangle such that its weight (the sum of weights of three sides) is negative.
\end{definition}

\subsection{Least Weight Subsequence}

We formally define $\LWS$ using the definition provided in \cite{KPS17}. 

\begin{definition}[$\LWS$]
Consider a sequence of $n$ data items $x_1, \ldots, x_n$ and a weight matrix $w$ of size $n \times n$ where $w[i,j] \in \{-W,\ldots,W, \infty\}$ for all $1 \leq i< j \leq n$. The $\LWS$ problem asks to determine $T[n]$, which is defined by the following DP formulation:
\begin{align*}
    T[j]
    &=
    \begin{cases}
        0
        &
        \text{if } j = 0
        \\
        \displaystyle\min_{0 \leq i<j} \Bigl\{T[i]+w[i,j]\Bigr\}
        &
        \text{otherwise.}
    \end{cases}
\end{align*}
\end{definition}

Given a sequence of $n$ items, $\LWS$ computes a subsequence of those items which minimizes the total weight from the items chosen. We assume all the entries of $w$ can be accessed in $O(1)$ time. Figure \ref{fig: LWS} captures the idea of $\LWS$.

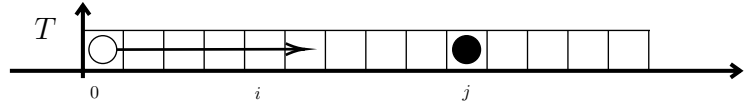
\begin{figure}[ht]

\begin{center}
\resizebox{10cm}{!}{

\tikzset{every picture/.style={line width=0.75pt}} 

\begin{tikzpicture}[x=0.75pt,y=0.75pt,yscale=-1,xscale=1]

\draw [line width=2.25]  (50,604.1) -- (592,604.1)(104.2,555.5) -- (104.2,609.5) (585,599.1) -- (592,604.1) -- (585,609.1) (99.2,562.5) -- (104.2,555.5) -- (109.2,562.5)  ;
\draw  [draw opacity=0] (104.2,574.1) -- (525.1,574.1) -- (525.1,605) -- (104.2,605) -- cycle ; \draw   (104.2,574.1) -- (104.2,605)(134.2,574.1) -- (134.2,605)(164.2,574.1) -- (164.2,605)(194.2,574.1) -- (194.2,605)(224.2,574.1) -- (224.2,605)(254.2,574.1) -- (254.2,605)(284.2,574.1) -- (284.2,605)(314.2,574.1) -- (314.2,605)(344.2,574.1) -- (344.2,605)(374.2,574.1) -- (374.2,605)(404.2,574.1) -- (404.2,605)(434.2,574.1) -- (434.2,605)(464.2,574.1) -- (464.2,605)(494.2,574.1) -- (494.2,605)(524.2,574.1) -- (524.2,605) ; \draw   (104.2,574.1) -- (525.1,574.1)(104.2,604.1) -- (525.1,604.1) ; \draw    ;
\draw  [color={rgb, 255:red, 0; green, 0; blue, 0 }  ,draw opacity=1 ][fill={rgb, 255:red, 0; green, 0; blue, 0 }  ,fill opacity=1 ] (378.67,588.33) .. controls (378.67,582.63) and (383.29,578) .. (389,578) .. controls (394.71,578) and (399.33,582.63) .. (399.33,588.33) .. controls (399.33,594.04) and (394.71,598.67) .. (389,598.67) .. controls (383.29,598.67) and (378.67,594.04) .. (378.67,588.33) -- cycle ;
\draw  [color={rgb, 255:red, 0; green, 0; blue, 0 }  ,draw opacity=1 ][fill={rgb, 255:red, 255; green, 255; blue, 255 }  ,fill opacity=1 ] (108.67,588.33) .. controls (108.67,582.63) and (113.29,578) .. (119,578) .. controls (124.71,578) and (129.33,582.63) .. (129.33,588.33) .. controls (129.33,594.04) and (124.71,598.67) .. (119,598.67) .. controls (113.29,598.67) and (108.67,594.04) .. (108.67,588.33) -- cycle ;
\draw [color={rgb, 255:red, 0; green, 0; blue, 0 }  ,draw opacity=1 ][line width=1.5]    (129.33,588.33) -- (266,588.01) ;
\draw [shift={(269,588)}, rotate = 179.86] [color={rgb, 255:red, 0; green, 0; blue, 0 }  ,draw opacity=1 ][line width=1.5]    (14.21,-4.28) .. controls (9.04,-1.82) and (4.3,-0.39) .. (0,0) .. controls (4.3,0.39) and (9.04,1.82) .. (14.21,4.28)   ;

\draw (384.33,613.07) node [anchor=north west][inner sep=0.75pt]    {$j$};
\draw (230.33,615.07) node [anchor=north west][inner sep=0.75pt]  [color={rgb, 255:red, 0; green, 0; blue, 0 }  ,opacity=1 ]  {$i$};
\draw (66.67,563.4) node [anchor=north west][inner sep=0.75pt]  [font=\LARGE]  {$T$};
\draw (108,614.07) node [anchor=north west][inner sep=0.75pt]    {$0$};

\end{tikzpicture}

}
\end{center}

\caption{$\LWS$. To compute the value of $T'[j]$ (black circle), we start from $T[0]$ and go through all possible $T[i]$ such that $1 \leq i<j \leq n$ and takes the minimum of all possible values (plus their respective weight values $w$).}
\label{fig: LWS}
\end{figure}

\cite{KPS17} also defines a ``static" version of $\LWS$ which is central to their reductions.

\begin{definition}
$([\Static]\LWS)$ Fix an instance of $\LWS$ with matrix $w$. Given intervals $I = \{a,a+1,\ldots,a+N-1\}, J = \{a+N,a+N+1,\ldots,a+2N-1\}$, together with correctly computed values $T[i]$ for all $i \in I$, the $[\Static]\LWS$ problem asks to determine
\[
T'[j] = \min_{i \in I}\Bigl\{T[i]+w[i,j]\Bigr\} \qquad \textup{ for all } j\in J.
\] 
\end{definition}

$[\Static] \LWS$ is a parallel, batch version of $\LWS$ that applies the $\LWS$ recurrence relation to many values of j at once, rather than to just a single $j$ value at a time. Indeed, we can compute all $T'[j]$ values where $j \in J$ in parallel because each $T'[j]$ only depends on the $T[i]$ values where $i \in I$, not on any other $T'[j]$ values. 

We use the notation $T'$ to highlight that $T'[j]$ may not equal $T[j]$ since $T'[j]$ is computed with partial information (in $I$). Figure \ref{fig: [static] LWS} captures the idea of $[\Static]\LWS$.

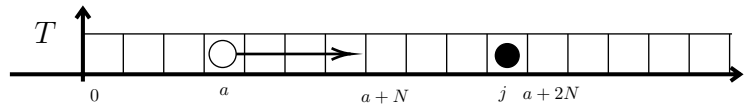
\begin{figure}[ht]
\begin{center}
\resizebox{10cm}{!}{

\tikzset{every picture/.style={line width=0.75pt}} 

\begin{tikzpicture}[x=0.75pt,y=0.75pt,yscale=-1,xscale=1]

\draw [line width=2.25]  (50,604.1) -- (592,604.1)(104.2,555.5) -- (104.2,609.5) (585,599.1) -- (592,604.1) -- (585,609.1) (99.2,562.5) -- (104.2,555.5) -- (109.2,562.5)  ;
\draw  [draw opacity=0] (104.2,574.1) -- (586,574.1) -- (586,605.1) -- (104.2,605.1) -- cycle ; \draw   (104.2,574.1) -- (104.2,605.1)(134.2,574.1) -- (134.2,605.1)(164.2,574.1) -- (164.2,605.1)(194.2,574.1) -- (194.2,605.1)(224.2,574.1) -- (224.2,605.1)(254.2,574.1) -- (254.2,605.1)(284.2,574.1) -- (284.2,605.1)(314.2,574.1) -- (314.2,605.1)(344.2,574.1) -- (344.2,605.1)(374.2,574.1) -- (374.2,605.1)(404.2,574.1) -- (404.2,605.1)(434.2,574.1) -- (434.2,605.1)(464.2,574.1) -- (464.2,605.1)(494.2,574.1) -- (494.2,605.1)(524.2,574.1) -- (524.2,605.1)(554.2,574.1) -- (554.2,605.1)(584.2,574.1) -- (584.2,605.1) ; \draw   (104.2,574.1) -- (586,574.1)(104.2,604.1) -- (586,604.1) ; \draw    ;
\draw  [color={rgb, 255:red, 0; green, 0; blue, 0 }  ,draw opacity=1 ][fill={rgb, 255:red, 0; green, 0; blue, 0 }  ,fill opacity=1 ] (410.33,590.17) .. controls (410.33,585.29) and (414.29,581.33) .. (419.17,581.33) .. controls (424.05,581.33) and (428,585.29) .. (428,590.17) .. controls (428,595.05) and (424.05,599) .. (419.17,599) .. controls (414.29,599) and (410.33,595.05) .. (410.33,590.17) -- cycle ;
\draw  [color={rgb, 255:red, 0; green, 0; blue, 0 }  ,draw opacity=1 ][fill={rgb, 255:red, 255; green, 255; blue, 255 }  ,fill opacity=1 ] (217.95,589.03) .. controls (217.95,583.49) and (213.46,579) .. (207.93,579) .. controls (202.39,579) and (197.9,583.49) .. (197.9,589.03) .. controls (197.9,594.56) and (202.39,599.05) .. (207.93,599.05) .. controls (213.46,599.05) and (217.95,594.56) .. (217.95,589.03) -- cycle ;
\draw [color={rgb, 255:red, 0; green, 0; blue, 0 }  ,draw opacity=1 ][line width=1.5]    (217.95,589.03) -- (300,589) ;
\draw [shift={(303,589)}, rotate = 179.98] [color={rgb, 255:red, 0; green, 0; blue, 0 }  ,draw opacity=1 ][line width=1.5]    (14.21,-4.28) .. controls (9.04,-1.82) and (4.3,-0.39) .. (0,0) .. controls (4.3,0.39) and (9.04,1.82) .. (14.21,4.28)   ;

\draw (411.33,611.07) node [anchor=north west][inner sep=0.75pt]    {$j$};
\draw (309.33,614.07) node [anchor=north west][inner sep=0.75pt]  [color={rgb, 255:red, 0; green, 0; blue, 0 }  ,opacity=1 ]  {$a+N$};
\draw (66.67,563.4) node [anchor=north west][inner sep=0.75pt]  [font=\LARGE]  {$T$};
\draw (108,614.07) node [anchor=north west][inner sep=0.75pt]    {$0$};
\draw (204.2,612.5) node [anchor=north west][inner sep=0.75pt]  [color={rgb, 255:red, 0; green, 0; blue, 0 }  ,opacity=1 ]  {$a$};
\draw (429.33,612.07) node [anchor=north west][inner sep=0.75pt]  [color={rgb, 255:red, 0; green, 0; blue, 0 }  ,opacity=1 ]  {$a+2N$};

\end{tikzpicture}

}
\end{center}

\caption{$[\Static] \LWS$. To compute $T'[j]$ (black circle), we take the minimum of all possible white circles from $T[a]$ to $T[a+N-1]$ (plus their respective weight values $w$).}
\label{fig: [static] LWS}
\end{figure}

\subsection{Higher Dimension $\LWS$}

We now define the core problems that this paper discusses again, which is a generalization of $\LWS$ to higher dimensions.

\begin{definition}[$\kLWS$]
Fix a positive integer $k$. Consider $(k+1)$-dimensional tensors $w_1,\ldots,w_{\ell}$ such that $w_{\ell}[i_1,\ldots,i_{k},j] \in \{-W,\ldots,W,\infty\}$ for all $1 \leq i_1,\ldots,i_k,j \leq n, 1 \leq \ell \leq k$. The $\kLWS$ problem asks to determine $T[n,\ldots,n]$ given the dynamic programming recurrence relation:
\begin{align*}
    T\Big[j_1,j_2,\ldots,j_k\Big] = 
    \begin{cases}
        0 \hspace{23.5em}\textup{ if } j_1 = j_2 = \ldots = j_k = 1 \\
        \displaystyle\min_{1\leq \ell \leq k}\Bigl\{\displaystyle\min_{1 \leq i_{\ell}<j_{\ell}}\Bigl\{T\Big[j_1,\ldots,j_{\ell-1},i_{\ell},j_{\ell+1},\ldots,j_k\Big]+w_{\ell}\Big[j_1,j_2,\ldots,j_k,i_{\ell}\Big]\Bigr\}\Bigr\} \textup{ otherwise.}
    \end{cases}
\end{align*}
\end{definition} An illustration of $\twoLWS$ can be found in \Cref{fig: kLWS}.

Like \cite{KPS17}, we also define $[\Static] \kLWS$, a generalization of $[\Static] \LWS$ to higher dimensions which is central to our reductions.

\begin{definition}
$([\Static]\kLWS)$ Given intervals $D_{a,a+N},D_{a+N,a+2N}$ together with correctly computed values $T_{\ell}[i_1,\ldots,i_k]$ for all $1 \leq \ell \leq k$ and $(i_1,\ldots,i_k) \in D_{a,a+N}$, $[\Static]\kLWS$ asks to determine 
\begin{align*}
    T'\Big[j_1,\ldots,j_k\Big] = \min_{1\leq \ell \leq k}\Bigg\{&\displaystyle\min_{a-I_{\ell}\leq i_{\ell}<a+N-I_{\ell}}\Bigl\{T_{\ell}\Big[j_1,\ldots,j_{\ell-1},i_{\ell},j_{\ell+1},\ldots,j_k\Big]+w_{\ell}\Big[j_1,j_2,\ldots,j_k,i_{\ell}\Big]\Bigr\}\Bigg\}  
\end{align*} for all $(j_1,j_2,\ldots,j_k) \in D_{a+N,a+2N}$.
\end{definition}

An illustration of $[\Static]\twoLWS$ can be found in \Cref{fig: static kLWS}.

\subsection{Tensor Ranks}

We give definitions of rank and slice rank for tensors. For notational convenience, we say that a problem has rank $d$ if its associated array/tensor has rank $d$.

\begin{definition}[Rank]
We say that a $k$-dimension array $w$ has \textup{rank} $d$ if there exists $k$ sets 
\[
X_1 = \{x_{11},\ldots,x_{1n}\},\ldots,X_k = \{x_{k1},\ldots,x_{kn}\}
\] of vectors with length $d$ such that 
$w[i_1,\ldots,i_k] = \langle x_{1,i_1},x_{2,i_2},\ldots,x_{k,i_ {k}} \rangle$ for all $1 \leq i_1,\ldots,i_k \leq n$.
\end{definition}

\begin{definition}[Slice Rank]
A $k$-dimensional (order-$k$) tensor $w \in \R^{n_1 \times n_2 \times \cdots \times n_d}$ has slice rank $1$ if there is a $j \in [d]$, a vector $a \in \R^{n_j}$, and a $(d-1)$-dimensional tensor $b \in \R^{n_1 \times \cdots \times n_{j-1} \times n_{j+1} \times \cdots \times n_d}$ such that, for all $i_1 \in [n_1], \ldots, i_d \in [n_d]$ we have 
\[
w[i_1, \ldots, i_d] = a[i_j] \cdot b[i_1, \ldots, i_{j-1}, i_{j+1}, \ldots, i_d].
\] More generally, the slice rank of tensor $w$ is the minimum non-negative integer $k$ such that there are slice rank $1$ tensors $w_1, \ldots, w_k$ for which $w = w_1 + \cdots + w_k$.
\end{definition}

\subsection{Polygon Triangulation}

In computational geometry, a polygon triangulation is a partition of a polygon $P$ into triangles. It is known that the number of partitions of a convex polygon is counted by the Catalan numbers \cite{stanley2015catalan}. We discuss the problem of finding the triangulation that minimizes the sum of weights of all triangles.

\begin{definition}[Polygon Triangulation]
Let $P(n)$ denote an $n$-sided convex polygon and fix an ordering of the polygon vertices. A triangulation of $P(n)$ is a partition of $P(n)$ into disjoint triangles using straight, internal edges between pair of nodes of $P(n)$. For each triangle $(v_i,v_j,v_k)$, let $w(i,j,k)$ be its weight, and the weight of a partition is the sum of all weights of its triangles. The polygon triangulation problem asks to determine the minimal weight of all partitions.
\end{definition}

For an example, consider the polygon triangulation of a polygon $P$ with 10 sides (see figure \ref{fig: polgyon triangulation}). Starting from the side $(1,10)$, we choose a node $6$ and partitions $P$ into 3 parts $P(1,6)$, triangle $(1,6,10)$ and $P(6,10)$, denoted by the black dashed lines. We further partitions $P(i,j)$ by choosing a node $i<k<j$ and partitions it to $P(i,k)$, triangle $(i,j,k)$ and $P(k,j)$. 

\begin{figure}[ht]
\begin{center}
\resizebox{5cm}{!}{

\tikzset{every picture/.style={line width=0.75pt}} 

\begin{tikzpicture}[x=0.75pt,y=0.75pt,yscale=-1,xscale=1]

\draw  [line width=1.5]  (373,246.5) -- (346.93,326.73) -- (278.68,376.32) -- (194.32,376.32) -- (126.07,326.73) -- (100,246.5) -- (126.07,166.27) -- (194.32,116.68) -- (278.68,116.68) -- (346.93,166.27) -- cycle ;
\draw [line width=1.5]  [dash pattern={on 5.63pt off 4.5pt}]  (194.32,116.68) -- (278.68,376.32) ;
\draw [line width=1.5]  [dash pattern={on 5.63pt off 4.5pt}]  (278.68,116.68) -- (278.68,376.32) ;
\draw [color={rgb, 255:red, 208; green, 2; blue, 27 }  ,draw opacity=1 ][line width=1.5]  [dash pattern={on 5.63pt off 4.5pt}]  (194.32,116.68) -- (100,246.5) ;
\draw [color={rgb, 255:red, 208; green, 2; blue, 27 }  ,draw opacity=1 ][line width=1.5]  [dash pattern={on 5.63pt off 4.5pt}]  (100,246.5) -- (278.68,376.32) ;
\draw [color={rgb, 255:red, 74; green, 144; blue, 226 }  ,draw opacity=1 ][line width=1.5]  [dash pattern={on 5.63pt off 4.5pt}]  (100,246.5) -- (194.32,376.32) ;

\draw (185,82.4) node [anchor=north west][inner sep=0.75pt]  [font=\Large]  {$1$};
\draw (102,133.4) node [anchor=north west][inner sep=0.75pt]  [font=\Large]  {$2$};
\draw (68,228.4) node [anchor=north west][inner sep=0.75pt]  [font=\Large]  {$3$};
\draw (95,324.4) node [anchor=north west][inner sep=0.75pt]  [font=\Large]  {$4$};
\draw (178,389.4) node [anchor=north west][inner sep=0.75pt]  [font=\Large]  {$5$};
\draw (277,391.4) node [anchor=north west][inner sep=0.75pt]  [font=\Large]  {$6$};
\draw (361,323.4) node [anchor=north west][inner sep=0.75pt]  [font=\Large]  {$7$};
\draw (392,231.4) node [anchor=north west][inner sep=0.75pt]  [font=\Large]  {$8$};
\draw (362,139.4) node [anchor=north west][inner sep=0.75pt]  [font=\Large]  {$9$};
\draw (276,80.4) node [anchor=north west][inner sep=0.75pt]  [font=\Large]  {$10$};

\end{tikzpicture}

}
\end{center}

\caption{An example polygon triangulation problem.}
\label{fig: polgyon triangulation}
\end{figure}
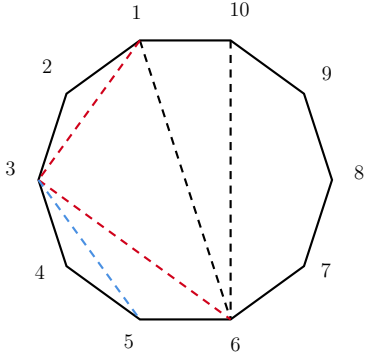

The polygon triangulation problem can be solved via dynamic programing, motivating our $\twoLWS^{\PT}$ problem.

\begin{definition}[$\twoLWS^{\PT}$]
Fix a tensor $w$. The $\twoLWS^{\PT}$ problem asks to compute the value of $T[n,n]$ given the dynamic programming recurrence relation:
\begin{align*}
    T[i,j] =
    \begin{cases}
        0 &\textup{ if } i+j \leq n+2\\
        \displaystyle\min_{0<k<i+j-n}\Big\{T[n-j+k,j]+T[i,j-k]+w[i,j,k]\Big\} &\textup{ otherwise}.
    \end{cases}
\end{align*} Under a change of variables/coordinates, this problem is equivalent to computing the value of $T[1,n]$ given the dynamic programming recurrence relation:
\begin{align*}
    T[i,j] = 
    \begin{cases}
        0 & \textup{ if } j-i \leq 1 \\
        \displaystyle\min_{i<k<j}\Big\{T[i,k]+T[k,j]+w[i,j,k]\Big\}&\textup{otherwise.}
    \end{cases}
\end{align*}
\end{definition}

It is not hard to see that polygon triangulation and $\twoLWS^{\PT}$ are the same problem: let $T[i,j]$ denote the weight of the sub-polygon containing nodes $i$ to $j$ and $w[i,j,k]$ be the weight of triangle $(v_i,v_j,v_k)$. More generally, any problem which splits an interval $[i, j]$ at some point $k$ where $k$ is between $i$ and $j$ can be understood as a $\twoLWS^{\PT}$ instance. Figure \ref{fig: 2D LWS^PT} captures the idea of $\twoLWS^{\PT}$.

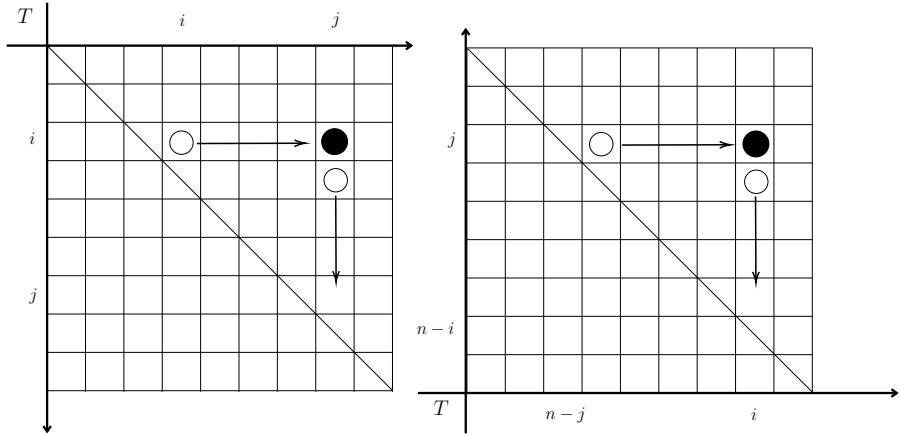
\begin{figure}[ht]
\begin{center}
\resizebox{12cm}{!}{

\tikzset{every picture/.style={line width=0.75pt}} 

\begin{tikzpicture}[x=0.75pt,y=0.75pt,yscale=-1,xscale=1]

\draw  [draw opacity=0] (76.95,461.72) -- (527.5,461.72) -- (527.5,913) -- (76.95,913) -- cycle ; \draw   (76.95,461.72) -- (76.95,913)(126.95,461.72) -- (126.95,913)(176.95,461.72) -- (176.95,913)(226.95,461.72) -- (226.95,913)(276.95,461.72) -- (276.95,913)(326.95,461.72) -- (326.95,913)(376.95,461.72) -- (376.95,913)(426.95,461.72) -- (426.95,913)(476.95,461.72) -- (476.95,913)(526.95,461.72) -- (526.95,913) ; \draw   (76.95,461.72) -- (527.5,461.72)(76.95,511.72) -- (527.5,511.72)(76.95,561.72) -- (527.5,561.72)(76.95,611.72) -- (527.5,611.72)(76.95,661.72) -- (527.5,661.72)(76.95,711.72) -- (527.5,711.72)(76.95,761.72) -- (527.5,761.72)(76.95,811.72) -- (527.5,811.72)(76.95,861.72) -- (527.5,861.72)(76.95,911.72) -- (527.5,911.72) ; \draw    ;
\draw [line width=2.25]  (76.95,405.82) -- (76.95,964.82)(551.43,461.72) -- (24.23,461.72) (81.95,957.82) -- (76.95,964.82) -- (71.95,957.82) (544.43,456.72) -- (551.43,461.72) -- (544.43,466.72)  ;
\draw    (76.95,461.72) -- (526.5,911) ;
\draw [color={rgb, 255:red, 0; green, 0; blue, 0 }  ,draw opacity=1 ][line width=1.5]    (272,589.38) -- (410.5,589.01) ;
\draw [shift={(413.5,589)}, rotate = 179.85] [color={rgb, 255:red, 0; green, 0; blue, 0 }  ,draw opacity=1 ][line width=1.5]    (14.21,-4.28) .. controls (9.04,-1.82) and (4.3,-0.39) .. (0,0) .. controls (4.3,0.39) and (9.04,1.82) .. (14.21,4.28)   ;
\draw [color={rgb, 255:red, 0; green, 0; blue, 0 }  ,draw opacity=1 ][line width=1.5]    (453,657.38) -- (453.49,767) ;
\draw [shift={(453.5,770)}, rotate = 269.75] [color={rgb, 255:red, 0; green, 0; blue, 0 }  ,draw opacity=1 ][line width=1.5]    (14.21,-4.28) .. controls (9.04,-1.82) and (4.3,-0.39) .. (0,0) .. controls (4.3,0.39) and (9.04,1.82) .. (14.21,4.28)   ;
\draw  [fill={rgb, 255:red, 255; green, 255; blue, 255 }  ,fill opacity=1 ] (236.74,588.25) .. controls (236.74,579.83) and (243.57,573) .. (251.99,573) .. controls (260.42,573) and (267.24,579.83) .. (267.24,588.25) .. controls (267.24,596.67) and (260.42,603.5) .. (251.99,603.5) .. controls (243.57,603.5) and (236.74,596.67) .. (236.74,588.25) -- cycle ;
\draw  [color={rgb, 255:red, 0; green, 0; blue, 0 }  ,draw opacity=1 ][fill={rgb, 255:red, 0; green, 0; blue, 0 }  ,fill opacity=1 ] (434.65,586.51) .. controls (434.88,577.18) and (442.62,569.81) .. (451.95,570.04) .. controls (461.27,570.26) and (468.65,578.01) .. (468.42,587.33) .. controls (468.19,596.66) and (460.45,604.04) .. (451.12,603.81) .. controls (441.8,603.58) and (434.42,595.84) .. (434.65,586.51) -- cycle ;
\draw [line width=2.25]  (560,914.72) -- (1184,914.72)(622.4,440.88) -- (622.4,967.38) (1177,909.72) -- (1184,914.72) -- (1177,919.72) (617.4,447.88) -- (622.4,440.88) -- (627.4,447.88)  ;
\draw  [fill={rgb, 255:red, 255; green, 255; blue, 255 }  ,fill opacity=1 ] (437.74,637.25) .. controls (437.74,628.83) and (444.57,622) .. (452.99,622) .. controls (461.42,622) and (468.24,628.83) .. (468.24,637.25) .. controls (468.24,645.67) and (461.42,652.5) .. (452.99,652.5) .. controls (444.57,652.5) and (437.74,645.67) .. (437.74,637.25) -- cycle ;
\draw  [draw opacity=0] (623.95,464.72) -- (1074.5,464.72) -- (1074.5,916) -- (623.95,916) -- cycle ; \draw   (623.95,464.72) -- (623.95,916)(673.95,464.72) -- (673.95,916)(723.95,464.72) -- (723.95,916)(773.95,464.72) -- (773.95,916)(823.95,464.72) -- (823.95,916)(873.95,464.72) -- (873.95,916)(923.95,464.72) -- (923.95,916)(973.95,464.72) -- (973.95,916)(1023.95,464.72) -- (1023.95,916)(1073.95,464.72) -- (1073.95,916) ; \draw   (623.95,464.72) -- (1074.5,464.72)(623.95,514.72) -- (1074.5,514.72)(623.95,564.72) -- (1074.5,564.72)(623.95,614.72) -- (1074.5,614.72)(623.95,664.72) -- (1074.5,664.72)(623.95,714.72) -- (1074.5,714.72)(623.95,764.72) -- (1074.5,764.72)(623.95,814.72) -- (1074.5,814.72)(623.95,864.72) -- (1074.5,864.72)(623.95,914.72) -- (1074.5,914.72) ; \draw    ;
\draw    (624.95,465.72) -- (1073.5,913) ;
\draw  [color={rgb, 255:red, 0; green, 0; blue, 0 }  ,draw opacity=1 ][fill={rgb, 255:red, 0; green, 0; blue, 0 }  ,fill opacity=1 ] (983.65,589.51) .. controls (983.88,580.18) and (991.62,572.81) .. (1000.95,573.04) .. controls (1010.27,573.26) and (1017.65,581.01) .. (1017.42,590.33) .. controls (1017.19,599.66) and (1009.45,607.04) .. (1000.12,606.81) .. controls (990.8,606.58) and (983.42,598.84) .. (983.65,589.51) -- cycle ;
\draw  [fill={rgb, 255:red, 255; green, 255; blue, 255 }  ,fill opacity=1 ] (985.74,639.25) .. controls (985.74,630.83) and (992.57,624) .. (1000.99,624) .. controls (1009.42,624) and (1016.24,630.83) .. (1016.24,639.25) .. controls (1016.24,647.67) and (1009.42,654.5) .. (1000.99,654.5) .. controls (992.57,654.5) and (985.74,647.67) .. (985.74,639.25) -- cycle ;
\draw [color={rgb, 255:red, 0; green, 0; blue, 0 }  ,draw opacity=1 ][line width=1.5]    (1000,658.38) -- (1000.49,768) ;
\draw [shift={(1000.5,771)}, rotate = 269.75] [color={rgb, 255:red, 0; green, 0; blue, 0 }  ,draw opacity=1 ][line width=1.5]    (14.21,-4.28) .. controls (9.04,-1.82) and (4.3,-0.39) .. (0,0) .. controls (4.3,0.39) and (9.04,1.82) .. (14.21,4.28)   ;
\draw  [fill={rgb, 255:red, 255; green, 255; blue, 255 }  ,fill opacity=1 ] (783.74,590.25) .. controls (783.74,581.83) and (790.57,575) .. (798.99,575) .. controls (807.42,575) and (814.24,581.83) .. (814.24,590.25) .. controls (814.24,598.67) and (807.42,605.5) .. (798.99,605.5) .. controls (790.57,605.5) and (783.74,598.67) .. (783.74,590.25) -- cycle ;
\draw [color={rgb, 255:red, 0; green, 0; blue, 0 }  ,draw opacity=1 ][line width=1.5]    (826,591.38) -- (964.5,591.01) ;
\draw [shift={(967.5,591)}, rotate = 179.85] [color={rgb, 255:red, 0; green, 0; blue, 0 }  ,draw opacity=1 ][line width=1.5]    (14.21,-4.28) .. controls (9.04,-1.82) and (4.3,-0.39) .. (0,0) .. controls (4.3,0.39) and (9.04,1.82) .. (14.21,4.28)   ;

\draw (52.5,573.4) node [anchor=north west][inner sep=0.75pt]  [font=\LARGE]  {$i$};
\draw (447,416.9) node [anchor=north west][inner sep=0.75pt]  [font=\LARGE]  {$j$};
\draw (247.83,419.73) node [anchor=north west][inner sep=0.75pt]  [font=\LARGE]  {$i$};
\draw (52.33,776.57) node [anchor=north west][inner sep=0.75pt]  [font=\LARGE]  {$j$};
\draw (37.5,411.9) node [anchor=north west][inner sep=0.75pt]  [font=\huge]  {$T$};
\draw (579.5,922.9) node [anchor=north west][inner sep=0.75pt]  [font=\huge]  {$T$};
\draw (992.83,933.73) node [anchor=north west][inner sep=0.75pt]  [font=\LARGE]  {$i$};
\draw (598,572.9) node [anchor=north west][inner sep=0.75pt]  [font=\LARGE]  {$j$};
\draw (725,932.28) node [anchor=north west][inner sep=0.75pt]  [font=\LARGE]  {$n-j$};
\draw (557,821.28) node [anchor=north west][inner sep=0.75pt]  [font=\LARGE]  {$n-i$};
\end{tikzpicture}

}
\end{center}
\caption{$\twoLWS^{\PT}$. To compute $T[i,j]$ (black circle), we are taking the minimum over the sum of all possible \textbf{pairs} of white circles (plus their respective weight values $w$). The solution to $\twoLWS^{\PT}$ is found at $T[1, n]$ in the left figure and at $T[n, n]$ in the right figure due to a change of variables.}
\label{fig: 2D LWS^PT}
\end{figure}

\section{$k$-dimensional Least Weight Subsequence $(\kLWS)$}
\label{sec: kLWS}

In this section, we discuss $\kLWS$ with rank and slice rank respectively. For $\kLWS$ with rank $d$, we prove the reductions in the following diagram. 

\[
\includegraphics[scale=0.15]{Pic01.png}
\] All our reductions preserve the rank of the problem, so this diagram shows that there exists truly sub-cubic algorithm for $\kLWS$ and $[\Static]\kLWS$ when the rank is constant (because $\MinIP$ does \cite{Williams18}). In addition, we show that $\kMinIP$ reduces to $\kLWS$ and $[\Static]\kLWS$, so their hardness are established. We delay our proof of $\SAT$ reducing to $\kMinIP$ to \Cref{sec: SAT to kMinIP} because it mimics the proof of $\SAT$ reducing to $\MinIP$ in \cite{Chen18}. 

In addition, we show that $\twoLWS,[\Static]\twoLWS$ with slice rank $3$ or above is $\APSP$-hard, and give truly sub-cubic algorithms for $\twoLWS,[\Static]\twoLWS$ with slice rank $1$.

\subsection{Rank $d$ $[\Static]\kLWS$ Hierarchy}
\label{sec: Static kLWS Hierarchy}

In this section we establish a hierarchy for $\kLWS$ and $[\Static]\kLWS$ with rank $d$.

\textbf{Notations}: 
\begin{itemize}
    \item $I_{j} = (\sum i_{\ell})-i_{j}$, $I_{j,t} = (\sum i_{\ell})-i_{j}-i_{t}$.
    \item $I_{j}' = (\sum i_{\ell}')-i_{j}'$, $I_{j,t}' = (\sum i_{\ell}')-i_{j}'-i_{t}'$.
    \item $D_{a,b}$ is the set of all $(i_1,\ldots,i_k)$ such that $a \leq i_1+\ldots+i_k<b$.
    \item $D_a = D_{a,a+1}$ is the set of all $(i_1,\ldots,i_k)$ such that $a = i_1+\ldots+i_k$.
\end{itemize}

\begin{theorem}
\label{thm: kMinIP to k-1 D LWS}
$(\kMinIP \rightarrow \kminusoneLWS)$ Suppose there exists an algorithm for $\kminusoneLWS$ with rank $d$ with running time $O(n^{k-\varepsilon})$ for some $\varepsilon>0$, then there exists an algorithm for $\kMinIP$ with rank $d$ with running time $O(n^{k-\delta})$ for some $\delta>0$. 
\end{theorem}
\begin{proof}
Given an $\kMinIP$ instance with 
\[
X_1 = \{x_{11},\ldots,x_{1n}\},\ldots,X_{k} = \{x_{k1},\ldots,x_{kn}\},
\] such that $x_{ij} \in \{-W,\ldots,W, \infty\}^{d}$ for all $1 \leq i\leq k, 1 \leq j \leq n$, we define $k$ sets of vectors 
\[
Y_1 = \{y_{11},\ldots,y_{1n}\},\ldots,Y_{k-1} = \{y_{k-1,1},\ldots,y_{k-1,n}\},Y_k = \{y_{k1},\ldots,y_{kn}\}
\] as follows: for all $1 \leq  \ell \leq k-1$,
\begin{align*}
    y_{\ell,j} = 
    \begin{cases}
        0^d & \textup{ if } 1 \leq j \leq (k-1)n\\
        x_{\ell, j\bmod n} &\textup{ if } (k-1)n+1 \leq j \leq kn.
    \end{cases} 
\end{align*} In addition, let
\begin{align*}
   y_{kj} = 
   \begin{cases}
       x_{k,j\bmod n} & \textup{ if } 1\leq j \leq (k-1)n\\
       0^d & \textup{ if } (k-1)n+1 \leq j \leq kn.
   \end{cases}
\end{align*} We claim that running $\kminusoneLWS_{n,d}$ algorithm with 
\[
w_{\ell}[i_1,\ldots,i_k] = \langle y_{1,i_{1}},\ldots,y_{k,i_{k}}\rangle
\] for all $\ell$ will give us $T[kn,\ldots,kn] = \min \langle x_{1,i_1},\ldots,x_{k,i_k} \rangle$. First notice that by our construction, when $(i_1,\ldots,i_{k-1}) \notin [(k-1)n+1,kn]^{k-1}$, we have $w_{\ell}[i_1,\ldots,i_k] = 0$. Therefore, $T[i_1,\ldots,i_k] = 0$ for all $(i_1,\ldots,i_k)$ such that $(i_1,\ldots,i_{k-1}) \notin [(k-1)n+1,kn]^{k-1}$. 

Now we use induction to show that 
\[
T\Big[(k-1)n+j_1,\ldots,(k-1)n+j_{k-1}\Big] = \min\Big\langle x_{1,i_1},\ldots,x_{k,i_k} \Big\rangle
\] where the minimum is taken over all $1 \leq i_1 \leq j_1,\ldots,1 \leq i_{k-1} \leq j_{k-1}$ and $1 \leq i_k \leq n$. This would suffice because if $j_{\ell} = n$ for all $1 \leq  \ell \leq k-1$, we would get the minimal inner product. The base case is when $j_{\ell} = 1$ for all $1 \leq \ell \leq k-1$. Then we have 
\begin{align*}
    &T\Big[(k-1)n+1,\ldots,(k-1)n+1\Big] \\
    &= \min_{1 \leq \ell \leq k-1} \Bigg\{\min_{1 \leq i_{\ell}<(k-1)n+1}\Big\{T\Big[(k-1)n+1,\ldots,(k-1)n+1,i_{\ell},(k-1)n+1,\ldots,(k-1)n+1\Big]+\\
    & \qquad w_{\ell}\Big[(k-1)n+1,\ldots,(k-1)n+1,i_{\ell}\Big]\Big\}\Bigg\}\\
    &= \min_{1 \leq \ell \leq k-1}\Bigg\{\min_{1 \leq i_{\ell}<(k-1)n+1}\Big\{w_{\ell}\Big[(k-1)n+1,\ldots,(k-1)n+1,i_{\ell}\Big]\Big\}\Bigg\}\\
    &= \min_{1 \leq \ell \leq n}\Big\langle x_{11},x_{21},\ldots,x_{k-1,1},x_{k,\ell} \Big\rangle.
\end{align*} For the induction step, we have
\begin{align*}
    &T\Big[(k-1)n+j_1,\ldots,(k-1)n+j_{k-1}\Big]\\
    &= \min_{1 \leq \ell \leq k-1} \Bigg\{\min_{1\leq i_{\ell}<(k-1)n+j_{\ell}}\Big\{T\Big[(k-1)n+j_{1},\ldots,(k-1)n+j_{\ell-1},(k-1)n+i_{\ell},(k-1)n+j_{\ell+1},\ldots,(k-1)n+j_{k-1}\Big]+\\
    & \qquad w_{\ell}\Big[(k-1)n+j_{1},\ldots,(k-1)n+j_{k-1},i_{\ell}\Big]\Big\}\Bigg\}\\
    &= \min_{1 \leq \ell \leq k-1} \Bigg\{\min_{1 \leq i_{\ell}<(k-1)n+1}\Big\{w_{\ell}\Big[(k-1)n+j_{1},\ldots,(k-1)n+j_{k-1},i_{\ell}\Big]\Big\},\\
    &\qquad \min_{(k-1)n+1 \leq i_{\ell} < (k-1)n+j_{\ell}}\Big\{T\Big[(k-1)n+j_{1},\ldots,(k-1)n+j_{\ell-1},(k-1)n+i_{\ell},(k-1)n+j_{\ell+1},\ldots,(k-1)n+j_{k-1}\Big]\Big\}\Bigg\}\\
    &= \min_{1\leq \ell \leq k-1}\Bigg\{\min_{1\leq \ell \leq n}\Big\langle x_{1,j_{1}},\ldots,x_{k-1,j_{k-1}},x_{k,\ell}\Big\rangle,\\
    &\qquad \min_{(k-1)n+1 \leq i_{\ell} < (k-1)n+j_{\ell}}\Big\{T\Big[(k-1)n+j_{1},\ldots,(k-1)n+j_{\ell-1},(k-1)n+i_{\ell},(k-1)n+j_{\ell+1},\ldots,(k-1)n+j_{k-1}\Big]\Big\}\Bigg\}.
\end{align*} By induction hypothesis, 
\[
\min_{(k-1)n+1 \leq i_{\ell} < (k-1)n+j_{\ell}}\Big\{T\Big[(k-1)n+j_{1},\ldots,(k-1)n+j_{\ell-1},(k-1)n+i_{\ell},(k-1)n+j_{\ell+1},\ldots,(k-1)n+j_{k-1}\Big]\Big\}
\] is the minimum over all $\langle x_{1,i_{1}},\ldots,x_{1,i_{k}} \rangle$ over
\[
(i_1,\ldots,i_k) \in [1,j_{1}]\times \ldots [1,j_{\ell-1}] \times [1,j_{\ell}-1] \times [1,j_{\ell+1}]\times \ldots \times [1,j_{k-1}] \times [1,n].
\] Thus we are taking the minimum over all $(i_1,\ldots,i_k)$ in
\[
\{j_1\} \times\ldots \times \{j_k\}\times [1,n]\bigcup_{1 \leq \ell \leq k-1}\Big([1,j_{1}]\times \ldots [1,j_{\ell-1}] \times [1,j_{\ell}-1] \times [1,j_{\ell+1}]\times \ldots \times [1,j_{k-1}] \times [1,n]\Big) = [1,j_1] \times \ldots \times [1,j_{k-1}]\times[1,n],
\] which concludes the induction.
\end{proof}

\begin{theorem}
\label{thm: kLWS to Static kLWS}
$(\kLWS \rightarrow [\Static]\kLWS)$ Suppose there exists an algorithm for $[\Static]\kLWS_{n,N,d}$ with running time $O(N^{2-\varepsilon}\cdot n^{k-1})$ for some $\varepsilon>0$, then there exists an algorithm for $\kLWS_{n,d}$ with running time $O(n^{k+1-\delta})$ for some $\delta>0$.  
\end{theorem}
\begin{proof}
Given an $\kLWS$ instance, we define a subproblem
\[
\mathcal{S}\Bigl(D_{\alpha,\beta},\Big\{t\Big[j_1,\ldots,j_k\Big]:(j_1,\ldots,j_k)\in D_{\alpha,\beta}\Big\}\Bigr)
\] as follows: Given $D_{\alpha,\beta}$ and $t[j_1,\ldots,j_k]$ for all $(j_1,\ldots,j_k)\in D_{\alpha,\beta}$ where
\begin{align*}
t\Big[j_1,\ldots,j_k\Big] = \min_{1 \leq \ell \leq k}\Big\{\min_{1\leq i_{\ell}<\alpha-J_{\ell}}\Big\{ T\Big[j_1,\ldots,j_{\ell-1},i_{\ell},j_{\ell+1},\ldots,j_k\Big]+w_{\ell}\Big[j_1,\ldots,j_k,i_{\ell}\Big]\Big\},\infty\Big\},
\end{align*} computes 
\[
T\Big[j_1,\ldots,j_k\Big] = \min_{1 \leq  \ell \leq k}\Big\{\min_{1 \leq i_{\ell}<j_{\ell}}\Big\{T\Big[j_1,\ldots,j_{\ell-1},i_{\ell},j_{\ell+1},\ldots,j_k\Big]+w_{\ell}\Big[j_1,\ldots,j_k,i_{\ell}\Big]\Big\}\Big\}
\] for all $(j_1,\ldots,j_k)\in D_{\alpha,\beta}$. Notice that a call 
\[
\mathcal{S}\Bigl(D_{k,kn},\Bigl\{t\Big[j_1,\ldots,j_k\Big]:(j_1,\ldots,j_k)\in D_{k,kn}\Bigr\}\Bigr)
\] with 
\begin{align*}
t\Big[j_1,\ldots,j_k\Big] &= 
\begin{cases}
    \displaystyle w_{\ell}\Big[j_1,\ldots,j_k,1\Big] & \textup{ if } j_{\ell} = 1 \textup{ for some }\ell\\
    \infty &\textup{otherwise}
\end{cases}
\end{align*} will solve the instance because only those who have a coordinate with $1$ will be assigned a value.

Now we solve $\mathcal{S}$ using Algorithm 1 below.
\begin{algorithm}
\label{alg: reduction to static}
\caption{$\mathcal{S} \rightarrow [\Static]\kLWS$}
        \If {$\alpha = \beta$}{\textbf{Return} $T\Big[j_1,j_2,\ldots,j_k\Big] = t\Big[j_1,j_2,\ldots,j_k\Big]$ for all $(j_1,j_2,\ldots,j_k) \in D_{\alpha}$}
        $m \leftarrow \Big\lceil \frac{\beta-\alpha}{2} \Big\rceil$
        \\
        $\Bigl\{T\Big[j_1,j_2,\ldots,j_k\Big] : j_1,j_2,\ldots,j_k \in D_{\alpha, \alpha + m }\Bigr\} \leftarrow S\Bigl(D_{\alpha, \alpha + m }, \Bigl\{t\Big[j_1,j_2,\ldots,j_k\Big] \in D_{\alpha, \alpha + m} \Bigr\}\Bigr)$
        \\
        Solve a $[\Static]\kLWS$ instance on $D_{\alpha, \alpha + m}, D_{\alpha+m, \beta}$ with correctly computed values
        \[
        \Big\{T\Big[j_1,j_2,\ldots,j_k\Big] : (j_1,j_2,\ldots,j_k) \in D_{\alpha, \alpha + m}\Big\}
        \] and output 
        \[
        \Big\{T'\Big[j_1,j_2,\ldots,j_k\Big] : (j_1,j_2,\ldots,j_k) \in D_{\alpha + m, \beta}\Big\}
        \]
        \\
        Let $t'\Big[j_1,j_2,\ldots,j_k\Big] = \min\Big\{t\Big[j_1,j_2,\ldots,j_k\Big], T'\Big[j_1,j_2,\ldots,j_k\Big]\Big\}$ for all $(j_1,j_2,\ldots,j_k) \in D_{\alpha+m, \beta}$
        \\
        $
        \Bigl\{T\Big[j_1,j_2,\ldots,j_k\Big] : (j_1,j_2,\ldots,j_k) \in  D_{\alpha+m, \beta }\Bigr\} \leftarrow S\Bigl(D_{\alpha+m, \beta}, \Bigl\{t'\Big[j_1,j_2,\ldots,j_k\Big] : (j_1,j_2,\ldots,j_k) \in  D_{\alpha+m, \beta} \Bigr\}\Bigr)
        $
        \\
        \If
        {
            $\beta - \alpha = 2m$
        }{
            \begin{align*}
                T\Big[j_1,j_2,\ldots,j_k\Big]
                =
                \min\Bigl\{t\Big[j_1,j_2,\ldots,j_k\Big],\min_{\alpha-J_{\ell}\leq i_{\ell}<\beta-J_{\ell}}\Big\{T\Big[j_1,\ldots,j_{\ell-1},i_{\ell},j_{\ell+1},\ldots,j_{k}\Big]+w\Big[j_1,\ldots,j_k,i_{\ell}\Big]\Big\}
                \Bigr\}
            \end{align*} for all $(j_1,j_2,\ldots,j_k) \in D_{\beta}$.
        }
        \textbf{Return} $\Big\{T\Big[j_1,j_2,\ldots,j_k\Big] : (j_1,j_2,\ldots,j_k) \in D_{\alpha, \beta}\Big\}$
\end{algorithm}


To see that the algorithm is correct, we use induction on $\beta-\alpha$. When $\alpha = \beta$ we want to compute $T[j_1,\ldots,j_k]$ for all $(j_1,\ldots,j_k) \in D_{\alpha}$, but by definition $t[j_1,\ldots,j_k] = T[j_1,\ldots,j_k]$ so we are done. 

Now line 4 of the algorithm correctly outputs $\Big\{T\Big[j_1,\ldots,j_k\Big]:(j_1,\ldots,j_k) \in D_{\alpha,\alpha+m-1}\Big\}$ because we input the correct $t\Big[j_1,\ldots,j_k\Big]$ and by induction hypothesis. In line 6 we compute for all $(j_1,\ldots,j_k) \in D_{\alpha+m,\beta-1}$:
\begin{align*}
    t'\Big[j_1,\ldots,j_k\Big] &= \min\Big\{t\Big[j_1,\ldots,j_k\Big],T'\Big[j_1,\ldots,j_k\Big]\Big\}\\
    &= \min\Bigg\{\min_{1 \leq \ell \leq k}\Big\{\min_{1 \leq i_{\ell}<\alpha-J_{\ell}}\Big\{T\Big[j_1,\ldots,j_{\ell-1},i_{\ell},j_{\ell+1},\ldots,j_{k}\Big]+w\Big[j_1,\ldots,j_k,i_{\ell}\Big]\Big\}\Big\},\\
    &\hspace{4em} \min_{1 \leq \ell \leq k}\Big\{\min_{\alpha-J_{\ell} \leq i_{\ell}<\alpha+m-J_{\ell}}\Big\{T\Big[j_1,\ldots,j_{\ell-1},i_{\ell},j_{\ell+1},\ldots,j_{k}\Big]+w\Big[j_1,\ldots,j_k,i_{\ell}\Big]\Big\}\Big\}\Bigg\}\\
    &= \min_{1 \leq \ell \leq k}\Big\{\min_{1\leq i_{\ell}<\alpha+m-J_{\ell}}\Big\{T\Big[j_1,\ldots,j_{\ell-1},i_{\ell},j_{\ell+1},\ldots,j_{k}\Big]+w\Big[j_1,\ldots,j_k,i_{\ell}\Big]\Big\}\Big\}.
\end{align*}

Therefore, these are the correct values for all $(j_1,\ldots,j_k) \in D_{\alpha+m,\beta-1}$ for applying $\mathcal{S}$ again in line 7, where we correctly output $\Big\{T\Big[j_1,\ldots,j_k\Big] : (j_1,\ldots,j_k) \in D_{\alpha + m, \beta-1} \Big\}$. Finally, if $\beta - \alpha = 2m$, our algorithm computes the minimum over all $i_{\ell}$ such that $(j_1,\ldots,j_{\ell-1},i_{\ell},j_{\ell+1},\ldots,j_k)\in D_{1,\beta-1}$, so we are outputting the correct values.

The runtime of our algorithm for $\mathcal{S}$, when $\beta-\alpha = N$, can be expressed as
\[
T_{\mathcal{S}}(n,N) \leq 2T_{\mathcal{S}}\Bigl(n,\frac{N}{2}\Bigr)+T_{\Static}\Bigl(n,\frac{N}{2}\Bigr)+O(N n^{k-1})
\] because we have 2 recursive calls with size $\frac{N}{2}$, run a $[\Static]$ algorithm with length $\frac{N}{2}$, and computes $t'[i_1,\ldots,i_k]$ for $O(N n^{k-1})$ values. By assumption we have $T_{\Static}\Bigl(\frac{N}{2}\Bigr) \leq O(N^{2-\varepsilon} n^{k-1})$, so  this recursive formula becomes
\[
T_{\mathcal{S}}(n,N) \leq 2T_{\mathcal{S}}\Bigl(n,\frac{N}{2}\Bigr)+O(N^{2-\varepsilon} n^{k-1}).
\] Solving it gives $T_{\mathcal{S}}(n,N) \leq O(N^{2-\varepsilon}\cdot n^{k-1}\cdot \log N) \leq O(N^{2-\delta}n^{k-1})$ for some $\delta>0$. Therefore, 
\[
T_{\mathcal{S}}(n,n) \leq O(n^{k+1-\delta}).
\]

\end{proof}

Since $\SAT$ reduces to $\kMinIP_{n,2^{O(\log^{*}n)}}$ (\Cref{sec: SAT to kMinIP}), we immediately get hardness results of $\kLWS$ and $[\Static]\kLWS$ from $\SETH$:
\begin{corollary}
Assuming $\SETH$, there is no $O(n^{k+1-\varepsilon})$ time algorithm for $\kLWS$ or $[\Static]\kLWS$ with rank at least $2^{O(\log^{*}n)}$ for any $\varepsilon>0$.
\end{corollary}

Finally, we show that just like $\kMinIP$, $[\Static]\kLWS$ also exhibits a hierarchy.

\begin{theorem}
\label{thm: static kLWS to static k-1 LWS}
$([\Static]\kLWS \rightarrow [\Static]\kminusoneLWS)$ Suppose there exists an algorithm for $[\Static]\kminusoneLWS_{n,N,d}$ with running time $O(N^{2-\varepsilon}\cdot n^{k-2})$ for some $\varepsilon>0$, then there exists an algorithm for $[\Static]\kLWS_{n,N,d}$ with running time $O(N^{2-\delta}\cdot n^{k-1})$ for some $\delta>0$.
\end{theorem}
\begin{proof} 
Given an $[\Static]\kLWS_{n,d}$ instance with $D_{a,a+N},D_{a+N,a+2N}$ together with correctly computed values $T[j_1,\ldots,j_k]$ for all $1 \leq \ell \leq k, (i_1,\ldots,i_k) \in D_{a,a+N}$, we want to compute 
\begin{align*}
    T\Big[j_1,\ldots,j_k\Big] = \min_{1 \leq  \ell \leq k}\Bigg\{\min_{a-J_{\ell}\leq i_{\ell}<a+N-J_{\ell}}\Big\{T\Big[j_1,\ldots,j_{\ell-1},i_{\ell},j_{\ell+1},\ldots,j_{k}\Big]+w_{\ell}\Big[j_1,\ldots,j_{k},i_{\ell}\Big]\Big\}\Bigg\} 
\end{align*} for all $(j_1,j_2,\ldots,j_k) \in D_{a+N,a+2N}$. Fix some $n-a-N \leq j \leq n$. For any $(j,j_2,\ldots,j_k) \in D_{a+N,a+2N}$ we have 
\begin{align*}
    T\Big[j,j_2,\ldots,j_k\Big] &= \min_{1 \leq  \ell \leq k}\Bigg\{\min_{a-J_{\ell}\leq i_{\ell}<a+N-J_{\ell}}\Big\{T\Big[j,j_2,\ldots,j_{\ell-1},i_{\ell},j_{\ell+1},\ldots,j_{k}\Big]+w_{\ell}\Big[j,j_2,\ldots,j_{k},i_{\ell}\Big]\Big\}\Bigg\}\\
    &= \min\Bigg\{\min_{2\leq \ell \leq k}\Big\{\min_{1\leq i_{\ell}<j_{\ell}}T\Big[j,j_2,\ldots,j_{\ell-1},i_{\ell-1},j_{\ell+1},\ldots,j_k\Big]+w_{\ell}\Big[j,j_2,\ldots,j_{\ell-1},i_{\ell},j_{\ell+1},\ldots,j_k\Big]\Big\},\\
    &\qquad \min_{1\leq i_1 < j_1}\Big\{T\Big[i_1,j_2,\ldots,j_k\Big]+w_1\Big[i_1,j_2,\ldots,j_k\Big]\Big\}
    \Bigg\}
\end{align*} We can compute the first term in the minimum using a $[\Static]\kminusoneLWS$ algorithm with time $O(N^{2-\varepsilon}\cdot n^{k-2})$ and the second term using a $\LWS$ algorithm with time at most $O(N^2\cdot n)$. This is because after we fix $j$, $w[j,\ldots]$ still has rank $d$ but one less dimension. Repeat this process for all the $j$ on all $k$ coordinates to solve $[\Static]\kLWS$, and the total running time is at most
\[
kn\cdot \Big(O(N^2\cdot n)+O(N^{2-\varepsilon}\cdot n^{k-2})\Big) = O(N^{2-\delta}\cdot n^{k-1})
\] for some $\delta>0$.
\end{proof}

\subsection{Slice Rank $\twoLWS$}
\label{sec: slice rank twoLWS}

In this section, we show that $\twoLWS,[\Static]\twoLWS$ with even slice rank $3$ is $\APSP$-hard, but $\twoLWS,[\Static]\twoLWS$ with slice rank $1$ is truly sub-cubic.

\begin{theorem}
\label{thm: APSP to twoLWS}
Assuming the $\APSP$ conjecture, there is no truly sub-cubic algorithm for $\twoLWS$ with slice rank $3$.
\end{theorem}
\begin{proof}
We reduce $\NegativeTriangle$ to $\twoLWS$ with slice rank $3$. Given an undirected graph $G = (V,E)$ where $V = \{v_1,\ldots,v_n\}$, we use $w$ to denote the weight function of an edge or a triangle. For convenience we let $w(v_a,v_a) = \infty$. We define both our tensors to be 
\[
\alpha[i,j,k] = f_1(i,k)\cdot g_1(j)+f_2(i,j)\cdot g_2(k)+f_3(k,j)\cdot g_3(i),
\] where
\[
f_1(i,k) = 
\begin{cases}
    w(v_{i-n},v_k) \textup{ if } i \in [n+1,2n], k \in [1,n]\\
    0 \hspace{1.5cm}\textup{ otherwise}
\end{cases},
g_1(j) = 
\begin{cases}
    1 \textup{ if } j \in [n+1,2n]\\
    0 \textup{ otherwise}
\end{cases}
\]
\[
f_2(i,j) = 
\begin{cases}
    w(v_{i-n},v_{j-n}) \textup{ if } i,j \in [n+1,2n]\\
    0 \hspace{2cm}\textup{ otherwise}
\end{cases},
g_2(k) = 
\begin{cases}
    1 \textup{ if } k \in [1,n]\\
    0 \textup{ otherwise}
\end{cases}
\]
\[
f_3(k,j) = 
\begin{cases}
    w(v_k,v_{j-n}) \textup{ if } k \in [1,n], j \in [n+1,2n]\\
    0 \hspace{1.5cm}\textup{ otherwise}
\end{cases},
g_3(i) = 
\begin{cases}
    1 \textup{ if } i \in [n+1,2n]\\
    0 \textup{ otherwise}.
\end{cases}
\] We claim that running $\twoLWS$ with $\alpha$ will solve the $\NegativeTriangle$ instance by $T[2n,2n]$ being the minimum weight of all triangles. In fact, we prove that for all $n+1 \leq i,j \leq 2n$, $T[i,j]$ is the minimum weight of all triangles $v_a,v_b,v_c$ such that $1 \leq a \leq i-n, 1 \leq b \leq j-n, 1 \leq c \leq n$.

Observe that:
\begin{itemize}
    \item When $i,j,k \in [1,n]$, we have $\alpha[i,j,k] = 0$. Therefore, $T[i,j] = 0$ for all $1 \leq i,j \leq n$.
    \item When $i \in [n]$ and $j \in [n+1,2n]$, $\alpha[i,j,k] = 0$. Therefore, $T[i,j] = 0$ when $i \in [n], j\in [n+1,2n]$.
    \item When $j \in [n]$ and $i \in [n+1,2n]$, $\alpha[i,j,k] = 0$. Therefore, $T[i,j] = 0$ when $j \in [n], j \in [n+1,2n]$.
    \item When $i,j \in [n+1,2n]$, $\alpha[i,j,k] = w(v_{i-n},b_{j-n},v_k)$ if $k \in [n]$ and $\alpha[i,j,k] = 0$ if $k \in [n+1,2n]$.
\end{itemize} Finally, we use induction on $(i,j)$ to prove the claim. When $i = j = n+1$, we have
\[
T[n+1,n+1] = \min_{1 \leq k \leq n}\Big\{\alpha[n+1,n+1,k]\Big\} = \min_{1\leq  k\leq n}\Big\{w(v_{i-n},v_{j-n},v_k)\Big\} = \infty.
\] When $i = n+2,j = n+1$, 
\[
T[n+2,n+1] = \min_{1 \leq k \leq n}\Big\{\alpha[n+2,n+1,k]\Big\} = \min_{1 \leq k \leq n}\Big\{w(v_2,v_1,v_k)\Big\}.
\] Similarly, when $i = n+1,j= n+2$, 
\[
T[n+1,n+2] = \min_{1\leq k \leq n}\{w(v_1,v_2,v_k)\}.
\] Now for general $(i,j) \in [n+1,2n]^2$, we have
\begin{align*}
    T[i,j] &= \min\Big\{\min_{1 \leq k<i}\Big\{T[k,j]+\alpha[i,j,k]\Big\},\min_{1\leq k <j}\Big\{T[i,k]+\alpha[i,j,k]\Big\}\Big\}\\
    &= \min\Big\{\min_{1 \leq k\leq n}\Big\{T[k,j]+\alpha[i,j,k]\Big\},\min_{n+1 \leq k<i}\Big\{T[k,j]+\alpha[i,j,k]\Big\},\min_{1\leq k \leq n}\Big\{T[i,k]+\alpha[i,j,k]\Big\},\\
    &\quad \min_{n+1\leq k <j}\Big\{T[i,k]+\alpha[i,j,k]\Big\}\Big\}\\
    &= \min\Big\{\min_{1\leq k\leq n}\Big\{w(v_{i-n},v_{j-n},v_k)\Big\},\min_{n+1 \leq k<i}\Big\{T[k,j]\Big\},\min_{n+1\leq k<j}\Big\{T[i,k]\Big\}\Big\}.
\end{align*} By induction hypothesis we know we are taking the minimum over all triangles $(v_a,v_b,v_c)$ such that 
\begin{align*}
(a,b,c)& \in [1,i-n-1] \times [1,j-n] \times [n] \bigcup [1,i-n]\times [1,j-n-1] \times [n] \bigcup \{i-n\}\times \{j-n\} \times [n] \\
&= [1,i-n] \times [1,j-n] \times [n].    
\end{align*} Therefore our claim is proved.

\end{proof}

Now we show that $[\Static]\twoLWS$ with slice rank $1$ can be solved in truly sub-cubic time (which implies $\twoLWS$ with slice rank $1$ is also truly sub-cubic) but there is no truly sub-cubic algorithm for $\twoLWS$ with slice rank $3$ assuming $\APSP$ conjecture.

\begin{theorem}
\label{thm: static twoLWS with slice rank 1 is truly sub-cubic}
$[\Static]\twoLWS_{n,N}$ with slice rank $1$ can be solved in $O(n^2\cdot N^{1-\varepsilon})$ for some $\varepsilon>0$.
\end{theorem}
\begin{proof}
The idea is similar to the proof of \Cref{thm: static kLWS to static k-1 LWS}. We reduce $[\Static]\twoLWS$ with slice rank $1$ to $[\Static]\LWS$ with rank $1$, for which we know there exists an $O(N^{1-\varepsilon}\cdot n)$ time algorithm. 

Given an instance of $[\Static]\twoLWS$ with slice rank $1$, tensors $w_1,w_2$, there are 3 possibilities for $w_1$:
\begin{itemize}
    \item $w_1[i,j,k] = f(i,k)\cdot g(j)$ for some $f,g$. In this case, we fix an $i \in D_{a+N,a+2N}$ such that $w_1[i,j,k]$ becomes a matrix with rank $2$. Now we can run $[\Static]\LWS$ algorithm with rank $2$ (at most $n$ times) to compute 
    \[
    \min_{1 \leq k'<j'} \Big\{T[i,k']+w_1[i,j',k']\Big\}
    \] for all $(i,j') \in D_{a+N,a+2N}$ in time $O(N^{1-\varepsilon}n)$ for some $\varepsilon>0$. Doing this for all $i$, we can compute  
    \[
    \min_{1 \leq k<j} \Big\{T[i,k]+w_1[i,j,k]\Big\}
    \] for all $(i,j) \in D_{a+N,a+2N}$ in time at most $O(N^{1-\varepsilon}\cdot n^2)$.

    \item $w_1[i,j,k] = f(i,j)\cdot g(k)$ for some $f,g$. This is similar to the previous case but instead of fixing $i$, we fix $j$. Similarly, we can compute 
    \[
    \min_{1\leq k<j}\Big\{T[i,k]+w_1[i,j,k]\Big\}
    \] for all $(i,j)\in D_{a+N,a+2N}$ in time $O(N^{1-\varepsilon}\cdot n^2)$. 
    
    \item $w_1[i,j,k] = f(j,k)\cdot g(i)$ for some $f,g$. Again we can run $[\Static]\LWS$ for each $j$ and compute
    \[
    \min_{1\leq k<j}\Big\{T[i,k]+w_1[i,j,k]\Big\}
    \] for all $(i,j)\in D_{a+N,a+2N}$ in time $O(N^{1-\varepsilon}\cdot n^2)$. 
\end{itemize} The analysis for $w_2[i,j,k]$ is the same, and thus we can compute
\[
\min_{1\leq k<i}\Big\{T[k,j]+w_2[i,j,k]\Big\}
\] for all $(i,j)\in D_{a+N,a+2N}$ in time $O(N^{1-\varepsilon}\cdot n^2)$. Finally take the pairwise minimum to give
\[
\min\Big\{\min_{1\leq k<j}\Big\{T[i,k]+w_1[i,j,k]\Big\},\min_{1\leq k<i}\Big\{T[k,j]+w_2[i,j,k]\Big\}\Big\}
\] for all $(i,j)\in D_{a+N,a+2N}$.
\end{proof}

\begin{theorem}
\label{thm: twoLWS slice rank 1 is truly sub-cubic}
$\twoLWS$ with slice rank 1 is truly sub-cubic.    
\end{theorem}
\begin{proof}
Immediately follows from \Cref{thm: static twoLWS with slice rank 1 is truly sub-cubic} and the fact that our reduction from $\twoLWS$ to $[\Static]\twoLWS$ in \Cref{thm: kLWS to Static kLWS} preserves slice rank.
\end{proof}

\section{Polygon Triangulation $(\twoLWS^{\PT})$}

\label{sec: polygon triangulation}

In this section, we discuss the polygon triangulation problem $\twoLWS^{\PT}$ and its connections with $\twoLWS$. It was shown in \cite{HS02,HS82,HS84,LG21}, that if $w[i,j,k] = x_i\cdot x_j\cdot x_k$ for all $i,j,k$ with $x_i,x_j,x_k >0$ for all $i$, then $\twoLWS^{\PT}$ can be solved in $O(n^2)$ time. We establish several conditional hardness results of $\twoLWS^{\PT}$ based on $\SETH$ and $\APSP$ conjecture. Namely, $\twoLWS$ where $w_1 = w_2$ can be reduced to $\twoLWS^{\PT}$ with rank/slice rank unchanged, and thus finding the optimal triangulation for certain weight functions (rank $2^{O(\log^{*}n)}$ or slice rank $3$) is hard under $\SETH$. 

\subsection{Low Rank Polygon Triangulation is $\SETH$-hard}

\begin{theorem}
\label{thm: twoLWS reduces to PT}
$(\twoLWS\rightarrow \twoLWS^{\PT})$. There exists an $O(nd)$ time reduction from $\twoLWS_n$ with rank $d$, $w_1 = w_2 = w$, to $\twoLWS^{\PT}_{2n}$ with rank $d$.
\end{theorem}
\begin{proof}
Given an $\twoLWS_n$ instance $T$ with rank $d$ tensor $w[i,j,k] = \langle \mu_i,\sigma_j,\tau_k\rangle$ and
recurrence relation
\[
T[i,j] = \min\Big\{\min_{1\leq k<j}\Big\{T[i,k]+w[i,j,k]\Big\},\min_{i<k\leq n}\Big\{T[k,j]+w[i,j,k]\Big\}\Big\},
\] we want to compute $T[1,n]$. We construct an $\twoLWS^{\PT}_{2n}$ instance $T'$ as follows.
\[
\mu_i' = 
\begin{cases}
    \mu_i \textup{ if } i \in [1,n]\\
    0^d \textup{ if } i \in [n+1,2n],
\end{cases}
\sigma_j' = 
\begin{cases}
    0^d \textup{ if } j \in [1,n]\\
    \sigma_{j-n} \textup{ if } j \in [n+1,2n],
\end{cases}
\tau_k' = 
\begin{cases}
    \tau_k \textup{ if } k \in [1,n]\\
    \tau_{k-n} \textup{ if } k \in [n+1,2n],
\end{cases}
\] and define a $2n \times 2n \times 2n$ tensor as $w'[i,j,k]=\langle \mu_i',\sigma_j',\tau_k'\rangle$.
We make a few observations:
\begin{itemize}
    \item $w'[i,j,k] = 0$ when $(i,j) \notin [1,n]\times [n+1,2n]$, and thus $T'[i,j] = 0$ for all $(i,j) \notin [1,n] \times [n+1,2n]$.
    \item When $(i,j)\in [n]\times [n+1,2n]$, $w'[i,j,k] = w[i,j-n,k]$ for all $k \in [n]$ and $w'[i,j,k] = w[i,j-n,k-n]$ for all $k \in [n+1,2n]$.
\end{itemize} We now prove that for all $(i,j) \in [n]\times [n+1,2n]$, we have $T'[i,j] = T[i,j-n]$, which will suffice since when $i = 1, j = 2n$, we have $T'[1,2n] = T[1,n]$.

We proceed by induction. The base case is $T'[n,n+1] = 0$. Now for each $(i,j) \in [1,n] \times [n+1,2n]$, by induction hypothesis we have
\begin{align*}
    T'[i,j] &= \min_{i<k<j} \Big\{T'[i,k]+T'[k,j]+w'[i,j,k]\Big\}\\
    &= \min\Big\{\min_{i<k \leq n} \Big\{T'[i,k]+T'[k,j]+w'[i,j,k]\Big\},\min_{n+1 \leq k <j}\Big\{T'[i,k]+T'[k,j]+w'[i,j,k]\Big\}\Big\}\\
    &= \min\Big\{\min_{i<k \leq n} \Big\{T'[k,j]+\langle\mu_i',\sigma_j',\tau_k'\rangle\Big\},\min_{1 \leq k <j-n}\Big\{T'[i,k-n]+\langle\mu_i',\sigma_j',\tau_k'\rangle\Big\}\Big\}\\
    &= T[i,j-n].
\end{align*} The time for our reduction is $O(nd)$, so the proof is complete.
\end{proof}

\begin{corollary}
\label{cor: PT hardness rank}
Under $\SETH$, there is no truly sub-cubic algorithm for $\twoLWS^{\PT}$ with weight function whose rank is $2^{O(\log^{*}n)}$ or above.
\end{corollary}
\begin{proof}
When we reduce $\threeMinIP$ to $\twoLWS$ in \Cref{thm: kMinIP to k-1 D LWS}, the tensors $w_{\ell}$ that we use are all the same, so \Cref{thm: kMinIP to k-1 D LWS} immediately gives a reduction from $\threeMinIP$ to $\twoLWS$ with $w_1 = w_2$ (preserving rank), which further reduces to $\twoLWS^{\PT}$ (preserving rank) by \Cref{thm: twoLWS reduces to PT}.
\end{proof}

\subsection{Constant Slice Rank Polygon Triangulation is $\APSP$-hard}

In fact, the we can modify the reduction in \Cref{thm: twoLWS reduces to PT} such that it preserved slice rank as well.
\begin{theorem}
\label{thm: twoLWS reduces to PT slice}
$(\twoLWS\rightarrow \twoLWS^{\PT}$, slice rank version$)$ There exists an $O(nd)$ time reduction from $\twoLWS_n$ with slice rank $d$, $w_1 = w_2 = w$, to $\twoLWS^{\PT}_{2n}$ with slice rank $d$.
\end{theorem}
\begin{proof}
Given an $\twoLWS$ instance with $n \times n$ tensor $w$ with rank $d$, by the proof of \Cref{thm: twoLWS reduces to PT}, we know it suffices to construct a $2n \times 2n$ tensor $w'$ such that 
\begin{itemize}
    \item $w'[i,j,k] = 0$ when $(i,j) \notin [1,n]\times [n+1,2n]$.
    \item When $(i,j)\in [n]\times [n+1,2n]$, $w'[i,j,k] = w[i,j-n,k]$ for all $k \in [n]$ and $w'[i,j,k] = w[i,j-n,k-n]$ for all $k \in [n+1,2n]$.
\end{itemize} Now for each slice in $w$, there are three possibilities. For each case we convert the slice into a new slice with dimension $2n \times 2n \times 2n$ such that $w'$ is the sum of them. 
\begin{itemize}
    \item $f(i,k)\cdot g(j)$: let 
    \[
    f'(i,k) = 
    \begin{cases}
        f(i-n,k) \textup{ if } i \in [1,n], k \in [1,n]\\
        f(i-n,k-n) \textup{ if } i\in [1,n], k \in [n+1,2n]\\
        0 \textup{ otherwise}
    \end{cases},
    g'(j) = 
    \begin{cases}
        g(j) \textup{ if } j \in [n+1,2n]\\
        0 \textup{ otherwise}.
    \end{cases}
    \] As a result, $f'(i,k)\cdot g'(j) = f(i,k)\cdot g(j-n)$ when $i \in [1,n], j \in [n+1,2n]$ and $k \in [n]$, $f'(i,k)\cdot g'(j) = f(i,k-n)\cdot g(j-n)$ when $i,j \in [n+1,2n]$ and $k \in [n+1,2n]$, and $f'(i,k)\cdot g'(j) = 0$ if $(i,j)\notin [1,n] \times [n+1,2n]$. Thus it satisfies the conditions above. 
    \item $f(i,j)\cdot g(k)$: let 
    \[
    f'(i,j) = 
    \begin{cases}
        f(i,j-n) \textup{ if } i \in [1,n],j \in [n+1,2n]\\
        0 \textup{ otherwise}
    \end{cases},
    g'(k) = 
    \begin{cases}
        g(k) \textup{ if } k \in [n]\\
        g(k-n) \textup{ if } k \in [n+1,2n].
    \end{cases}
    \] As a result, $f'(i,j)\cdot g'(k) = f(i,j-n)\cdot g(k)$ when $(i,j) \in [1,n] \times [n+1,2n]$ and $f'(i,j)\cdot g'(k) = 0$ otherwise. Again it satisfies the conditions above.

    \item $f(j,k)\cdot g(i)$: let 
    \[
    f'(j,k) = 
    \begin{cases}
        f(j-n,k) \textup{ if } j \in [n+1,2n], k \in [n]\\
        f(j-n,k-n) \textup{ if } j \in [n+1,2n], k \in [n+1,2n]\\
        0 \textup{ otherwise}.
    \end{cases},
    g'(i) = 
    \begin{cases}
        g(i) \textup{ if } i \in [n]\\
        0 \textup{ otherwise}.
    \end{cases}
    \] When $(i,j) \in [n] \times [n+1,2n], k\in [n]$, $f'(j,k)\cdot g'(i) = f(j-n,k)\cdot g(i)$, and when $(i,j) \in [n] \times [n+1,2n], k\in [n+1,2n]$, $f'(j,k)\cdot g'(i) = f(j-n,k-n)\cdot g(i)$. When $(i,j) \notin [n] \times [n+1,2n]$, $f'(j,k)\cdot g'(i) = 0$. It satisfies the conditions above. 
\end{itemize} Therefore, the sum of these slices must also satisfy the conditions of $w'$ imposed in \Cref{thm: twoLWS reduces to PT}, and the reduction takes $O(nd)$ time. 

\end{proof}

\begin{corollary}
\label{cor: PT hardness slice rank}
Under $\APSP$ conjecture, there is no truly sub-cubic algorithm for $\twoLWS^{\PT}$ with weight function whose slice rank is $3$ or above.
\end{corollary}
\begin{proof}
When we reduce $\APSP$ to $\twoLWS$ in \Cref{thm: APSP to twoLWS}, the tensors $\alpha$ that we use are the same, so \Cref{thm: APSP to twoLWS} immediately gives a reduction from $\APSP$ to $\twoLWS$ with $w_1 = w_2$, which further reduces to $\twoLWS^{\PT}$ (preserving slice rank) by \Cref{thm: twoLWS reduces to PT}.    
\end{proof}

\section{Applications of $\kLWS$}
\label{sec: applications}

In \Cref{sec: kLWS}, we have shown that $\kLWS$ can solve $\kMinIP$ and $\APSP$ with different tensors. In this section we discuss more applications of $\kLWS$.



\subsection{Higher Dimension Airplane Refueling}
\label{sec: airplane refueling}

The airplane refueling problem was brought by \cite{HL87} as an example of $\LWS$. 

\begin{definition}[Airplane Refueling]
Suppose an airplane needs to fly between 2 given airports which are distance $R$ apart. Suppose there are $n-1$ different refueling stops at distance $x_1,\ldots,x_{n-1}$ from the departure point and all stops lie on the segment between departure and destination points. We can let $0 = x_0 < x_1 < \ldots < x_n = R$. The cost of flying $\ell$ miles is $(k-\ell)^2$ for some $k>0$ (we prefer flying close to $k$ miles), and the goal is to fly from departure point to arrival point while minimizing the cost.    
\end{definition}

It is not hard to see setting $w[i,j] = (x_j - x_i - k)^2$ in $\LWS$ will solve the problem since $T[j]$ is always the minimum cost of flying from $x_0$ to $x_j$. $w$ has rank $4$ because
\[
w[i,j] = x_j^2\cdot 1 + 1 \cdot x_i^2+(-2x_j)\cdot (x_i+2k)+k\cdot (2x_i+k),
\] and \cite{HL87} shows that airplane refueling can be solved in linear time.

In the real world, it is usually unlikely that all refueling stops are located on a single line. In addition, the plane can move in multiple directions. The higher dimension airplane refueling problem is motivated by these observations. 

\begin{definition}[Higher Dimension Airplane Refueling]
Suppose an airplane needs to fly between two given airports on a $k$-dimensional grid with $n$ points at each dimension. Each point in the grid represents a refueling stop, and the cost of flying from stop $(i_{1},\ldots,i_{\ell-1},j_{\ell},i_{\ell+1},\ldots,i_k)$ to $(i_1,\ldots,i_k)$ to  is $c(i_1,\ldots,i_k,j_{\ell})$. The problem asks the minimum cost of flying from $(1,\ldots,1)$ to $(n,\ldots,n)$.
\end{definition}

Notice that this is closer to real-world scenario where trains need to travel on railways, or we are driving in a city with well-organized roads.

Setting $w[i_1,\ldots,i_{k+1}] = c[i_1,\ldots,i_{k+1}]$ in $\kLWS$ will solve the problem because $T[i_1,\ldots,i_k]$ will always be the minimum cost of flying from $(1,\ldots,1)$ to $(i_1,\ldots,i_k)$. If we were to follow the cost function suggested in \cite{HL87}, then we have $c[i_1,\ldots,i_k,j_{\ell}] = (L-(i_{\ell}-j_{\ell}))^2$, which has constant rank and thus it can be solved in time $O(n^{k+1-\varepsilon})$ for some $\varepsilon>0$. 

Another natural scenario is that the cost of flying from $(i_1,\ldots,i_{\ell-1},j_{\ell},i_{\ell+1},i_k)$ to $(i_1,\ldots,i_k)$ only depends on $(i_1,\ldots,i_k)$. It mimics the scenario that the airplane is charged a fee upon arrival.

\begin{definition}[Arrival Fee Airplane Refueling]
Suppose an airplane needs to fly between two given airports on a $k$-dimensional grid with $n$ points at each dimension. Each point in the grid represents a refueling stop, and the cost of flying from stop $(i_{1},\ldots,i_{\ell-1},j_{\ell},i_{\ell+1},\ldots,i_k)$ to $(i_1,\ldots,i_k)$ to  is $c(i_1,\ldots,i_k)$. The problem asks the minimum cost of flying from $(1,\ldots,1)$ to $(n,\ldots,n)$.    
\end{definition}

In the arrival fee airplane refueling problem with dimension $k$, the tensor has slice rank $1$, so by \Cref{thm: static twoLWS with slice rank 1 is truly sub-cubic} it can be solved in time $O(n^{k+1-\varepsilon})$ for some $\varepsilon>0$.

\subsection{Multiple Nested Boxes}
\label{sec: multiple nested boxes}

Nested boxes problem, or box stacking problem, is a famous example with a DP solution. 

\begin{definition}[Nested Boxes]
Given $n$ boxes in $d$ dimension of size $(b_1,\ldots,b_d)$, find the longest chain such that each box fits into the next (without rotation). We say that box $a$ of size $(a_1,\ldots,a_d)$ fits into box $b$ of size $(b_1,\ldots,b_d)$ if $a_i \leq b_i$ for all $1 \leq i \leq d$.    
\end{definition}

\cite{KPS17} proves that nested boxes is sub-quadratic equivalent to the vector domination problem defined in \cite{ILPS14} and both can be solved by $\LWS$: sort the boxes by volumn in increasing order as $B_1,\ldots,B_n$ and set $w_{ij}$ to be $-1$ if $B_j$ contains $B_i$ and $0$ otherwise. 

It is natural to consider the case where we are allowed to have multiple locations to put the boxes, which motivates our multiple nested boxes problem.

\begin{definition}[Multiple Nested Boxes]
Given $n$ boxes in $d$ dimension of size $(b_1,\ldots,b_d)$ and $k$ piles, find the maximum number of boxes we can use such that each in each pile, each box fits into the next (without rotation). We say that box $a$ of size $(a_1,\ldots,a_d)$ fits into box $b$ of size $(b_1,\ldots,b_d)$ if $a_i \leq b_i$ for all $1 \leq i \leq d$.    
\end{definition}

\begin{theorem}
\label{thm: multiple nested boxes}
Multiple nested boxes with $k$ dimension can be solved in time $O(n^{k+1-\varepsilon})$ for some $\varepsilon>0$.   
\end{theorem}
\begin{proof}
We first sort all the boxes by their volume in increasing order $B_1,\ldots,B_n$, and let 
\begin{align*}
    w_{\ell}[i_1,\ldots,i_k,j] = 
    \begin{cases}
        -1 & \textup{ if } B_j \textup{ fits into } B_{i_{\ell}}\\
        0 & \textup{ otherwise.}
    \end{cases}
\end{align*} To see that this indeed solves multiple nested boxes, we claim that $-T[i_1,\ldots,i_k]$ is the maximum number of boxes over all assignments such that the rank of the outer box in $t$-th pile is at most $i_t$. This will solve the multiple nested boxes when $i_1 = \ldots = i_k = n+1$. 

We proceed by induction. When $i_1 = \ldots = i_k = 1$, we have $T[i_1,\ldots,i_k] = 0$. There is no box with rank $0$ so we cannot put any boxes. Now for general $(i_1,\ldots,i_k)$, by recurrence we have 
\[
T\Big[i_1,\ldots,i_k\Big] = \min_{1 \leq \ell \leq k}\Big\{\min_{1 \leq j_{\ell}<i_{\ell}}\Big\{T\Big[i_1,\ldots,i_{\ell-1},j_{\ell},i_{\ell+1},i_{k}\Big]+w_{\ell}\Big[i_1,\ldots,i_k,j_{\ell}\Big]\Big\}\Big\}.
\] For any assignment $(u_1,\ldots,u_k)$ to the piles 
such that the $t$-th pile outer box has rank $u_t \leq i_t$, it can be achieved from adding $B_{u_{\ell}}$ to the assignment $(u_1,\ldots,u_{\ell-1},u_{\ell}',u_{\ell+1},\ldots,u_k)$ with the guarantee that $B_{u_{\ell}'}$ fits inside $B_{u_{\ell}}$. This case is covered by the right-hand-side of the equation because by induction hypothesis,
\[
-T\Big[u_1,\ldots,u_{\ell-1},u_{\ell}',u_{\ell+1},\ldots,u_k\Big]-w_{\ell}\Big[u_1,\ldots,u_k,u_{\ell}'\Big]
\] is the maximum over the assignments under this procedure.

Therefore, right-hand-side of the equation is exactly all possible ways to achieve the assignment $(u_1,\ldots,u_k)$ such that $u_t \leq i_t$ for all $t$, so our $\kLWS$ instance indeed solves the multiple nested boxes with $k$ dimension. In addition, notice that $w_{\ell}$ only depends on its $\ell$-th and last coordinate, so it can be expressed as a matrix. The same reasoning in \Cref{thm: static twoLWS with slice rank 1 is truly sub-cubic} shows that it can be reduced to $[\Static]\kminusoneLWS$ with rank $1$, which implies that it can be solved in time $O(n^{k+1-\varepsilon})$ for some $\varepsilon>0$.

\end{proof}

\section{Applications of $\twoLWS^\PT$} \label{sec:PT}


In section \ref{sec: polygon triangulation}, we showed we can solve polygon triangulation and $\twoLWS^\PT$ with different tensors. In this section, we discuss more applications of $\twoLWS^\PT$.

\subsection{Matrix-Chain Multiplication}

The matrix-chain multiplication problem was introduced in \cite{Godbole73} and is defined as follows:

\begin{definition}[Matrix-Chain Multiplication]
\label{def: matrix-chain multiplication}
    Given a chain of $n$ matrices $A_1, \hdots, A_n$ where matrix $A_i$ has dimension $d_{i-1} \times d_{i}$, find the order of matrix multiplications which minimizes the number of scalar multiplications using the straightforward matrix multiplication algorithm.
\end{definition}

Recall that multiplying an $n \times m$ matrix by an $m \times p$ matrix using the straightforward algorithm uses $n \cdot m \cdot p$ scalar multiplications. Moreover, the order in which you multiply a chain of matrices determines the number of scalar multiplications performed. For instance, consider three matrices $A_1, A_2, A_3$ with dimensions $(10, 20)$, $(20, 30)$, and $(30, 40)$ respectively. Multiplying $(A_1, A_2) A_3$ takes $(10 \cdot 20 \cdot 30) + (10 \cdot 30 \cdot 40) = 18000$ scalar multiplications while multiplying $A_1 (A_2 A_3)$ takes $(20 \cdot 30 \cdot 40) + (10 \cdot 20 \cdot 40) = 32000$ scalar multiplications.

Matrix-chain multiplication is a $\twoLWS^\PT$ problem where $T[i, j]$ is the cost of multiplying matrices $A_i, \ldots, A_j$ and we want to find the $k$ which minimizes the cost of multiplying $(A_i, \ldots, A_k)$ by $(A_{k+1}, \ldots, A_j)$. Multiplying matrices $(A_i, \ldots, A_k)$ would result in a matrix of dimension $(d_{i-1}, d_k)$ and multiplying $(A_{k+1}, \ldots, A_j)$ would result in a matrix of dimension $(d_k, d_j)$. Thus $T[i, j]$ equals the cost of multiplying all matrices $A_i, \ldots, A_k$ (i.e. $T[i, k]$) and all matrices $A_{k+1}, \ldots, A_j$ (i.e. $T[k, j]$) plus the cost of multiplying those two resultant matrices together (i.e. $d_{i-1} d_k d_j$). Setting $w_1[i, j, k] = w_2[i, j ,k] = d_{i-1} d_k d_j$ would solve this problem.

Let us construct a vector $d = [d_0, d_1, \ldots, d_n]$ with the dimensions of our matrices $A_1, \ldots, A_n$. Then $w$ has a tensor rank of 1 because it can be represented as the product of different entries of $d$, namely $w[i, j, k] = d[i-1] \cdot d[j] \cdot d[k]$. Moreover, there exists an $O(n \log n)$ time algorithm for this problem \cite{HS82, HS84}. Corollary \ref{cor: PT hardness rank} helps explain why this speedup is possible.

\subsection{Optimal Binary Search Tree}

The optimal binary search tree construction problem was introduced in \cite{Knuth71, GM59} and is defined as follows:
\begin{definition}[Optimal Binary Search Tree]
    Given a sequence of $n$ distinct keys $h_1, \hdots, h_n$ in sorted order where the probability of accessing key $h_i$ is $p_i$, construct a binary search tree from these keys which minimizes the expected access time.
\end{definition}

This problem is a $\twoLWS^\PT$ instance where $T[i, j]$ is the minimum cost binary search tree with keys $h_i, \ldots, h_j$. We want to chose a key $h_k$ to be the root of the sub-tree containing keys $h_i, \ldots, h_j$ which minimizes the expected access time. The expected access time for a key $h_t$ is $p_t \cdot (d_t + 1)$, the key's probability times its depth in the tree (i.e. the number of times this item has been accessed). We can compute this quantity incrementally, adding the probability $p_t$ of key $h_t$ once at every level it appears in the tree, summing up $p_t$ a total of $d_t$ times. Thus the expected cost of accessing keys $h_i, \ldots, h_j$ is $w[i, j, k] = \sum_{t=i}^j p_t$.

$w$ has a slice rank of 1 because it can be written as $w[i, j, k] = a[k] \cdot b[i, j]$ where $a[k] = 1$ and $b[i, j] = \sum_{t=i}^j p_t$. This observation, together with Corollary \ref{cor: PT hardness slice rank}, recovers the known $O(n^2)$ time algorithm for this problem \cite{Yao80, Yao82}.

\bibliographystyle{alpha}
\bibliography{sample} 

\appendix
\section{Reduction from $\SAT$ to $\kMinIP$ with rank $2^{O(\log^{*}n)}$}
\label{sec: SAT to kMinIP}

In this section, we reduce $\SAT$ to $\kMinIP$ with rank $2^{O(\log^{*}n)}$. Our reduction proceeds as follows:
\[
\includegraphics[scale=0.2]{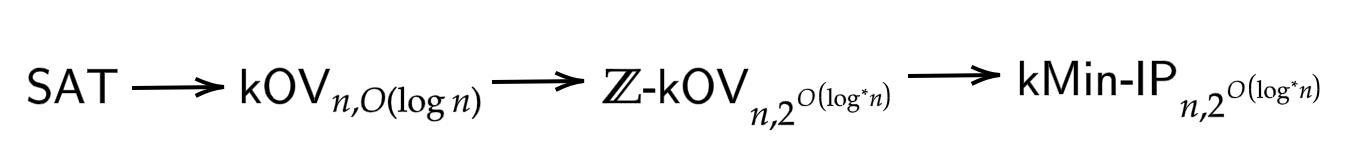}
\]

It is well-known that $\SAT$ reduces to $\OV$ \cite{Williams05}. The same proof can be used to give a reduction from $\SAT$ to $\kOV$, which we include here for completeness.
\begin{lemma}
$(\SAT \rightarrow \kOV_{n,O(\log n)})$. Suppose there exists an algorithm for $\kOV_{n,O(\log n)}$ with running time $O(n^{k-\varepsilon})$ for some $\varepsilon>0$, then there exists an algorithm for $\kSAT$ with running time $2^{(1-\delta)n}\cdot \poly(n)$ for some $\delta>0$, thus refuting $\SETH$.
\end{lemma}
\begin{proof}
For a $\kSAT$ formula $\varphi$ with $n$ variables and $m$ clauses, partition the variables $x_1,\ldots,x_n$ into $k$ subsets each with $n/k$ variables: 
\[
S_1 = \{x_1,\ldots,x_{\frac{n}{k}}\},\ldots,S_k = \{x_{n-\frac{n}{k}+1},\ldots,x_n\}.
\] For each $1 \leq i \leq k$, let $A_i = \{\alpha:S_i \rightarrow \{0,1\}\}$. Therefore, $|A_i| = 2^{n/k}$ for all $1 \leq i \leq k$. For each $\alpha \in A_i$, construct a vector $v_{\alpha} \in \{0,1\}^{m}$ such that 
\begin{align*}
V_{\alpha}[i] = 
\begin{cases}
    0 &\textup{ if $\alpha$ satisfies clause $i$}\\
    1 &\textup{otherwise.}
\end{cases}    
\end{align*} For each clause, $\alpha$ gives a partial assignment to it. If that assignment already makes the clause true, then $v_{\alpha}[i] = 0$. Let $Y_i = \{v_{\alpha}:\alpha \in A_i\}$, then $|Y_i| = 2^{n/k}$ for all $i$. Now notice that there exists $y_{i,j_i}$ such that $\langle y_{1,j_{1}},\ldots,y_{k,j_{k}} \rangle = 0$ if and only if $\varphi$ is satisfiable. 

Therefore running $\kOV_{n,O(\log n)}$ on $Y_1,\ldots,Y_k$ will tell us whether $\varphi$ is satisfiable or not. The running time of this algorithm is $(2^{n/k})^{k-\varepsilon}\cdot \poly(m)  =2^{(1-\delta)n}\cdot \poly(n)$.
\end{proof}

To give a reduction from $\kOV_{n,O(\log n)}$ to $\mathbb{Z}\textup{-}\kOV_{n,2^{O(\log^{*}n)}}$, we need the following results.

\begin{definition}[$\CRR$ encoding]
Let $b,\ell$ be two positive integers. For any vector $z \in \mathbb{Z}^{b}$ and a set of $b$ primes $q_i$, let $\CRR(z;\{q_i\})$ denote the unique integer in $\Bigl[\prod q_i\Bigr]$ such that $\CRR(\{z_i\};\{q_i\}) \equiv z_i \bmod q_i$ for all $1 \leq i\leq b$. For a vector $x \in \mathbb{Z}^{b\cdot \ell}$, partition $x$ into $\ell$ blocks, each of length $b$ and let $x^i$ be the $i$-th block. Define
\[
\CRR_{b,\ell}(x,\{q_i\}) = \Bigl(\CRR(x^1, \{q_i\}),\CRR(x^2, \{q_i\}),\ldots,\CRR(x^{\ell},\{q_i\})\Bigr).
\]
\end{definition}

\begin{theorem}
\label{thm: function}
Let $b,\ell$ be two sufficiently large integers. There is a reduction $\psi_{b,\ell}:\{0,1\}^{b\cdot \ell} \rightarrow \mathbb{Z}^{\ell}$ and a set $V_{b,\ell} \subset \mathbb{Z}$ such that for every $x_1,\ldots,x_k \in \{0,1\}^{b\cdot \ell}$, 
\[
\langle x_1,x_2,\ldots,x_k \rangle = 0 \Longleftrightarrow \Big\langle\psi_{b,\ell}(x_1),\ldots,\psi_{b,\ell}(x_k)\Big\rangle \in V_{b,\ell}
\] and $0 \leq \psi_{b,\ell}(x)_i < \ell^{(3k)^{\log^{*}b}\cdot b}$ for all possible $x$ and $i \in [\ell]$. Moreover, the computation of $\psi_{b,\ell}(x)$ takes $\poly(b\cdot \ell)$ time, and $V_{b,\ell}$ can be constructed in $O\Big(\ell^{O\big((3k)^{\log^{*}b}\cdot b\big)}\cdot\poly(b,\ell)\Big)$ time.
\end{theorem}
\begin{proof}
When $b$ is small (compared to $\ell$), we will define  $\psi_{b,\ell}$ directly; when $b$ is large, we will recursively define $\psi_{b,\ell}$.  

\textbf{Direct $\CRR$ for small $b$:} When $b < \ell$, we pick $b$ primes $q_1,\ldots,q_b$ in $[\ell+1,\ell^2]$, and we use them for our $\CRR$ encoding. Let 
\[
\psi_{b,\ell}(x) = \CRR_{b,\ell}(x,\{q_i\}) = \Bigl(\CRR(x^1, \{q_i\}),\CRR(x^2, \{q_i\}),\ldots,\CRR(x^{\ell},\{q_i\})\Bigr)
\] with $q_1,\ldots,q_b$ for any $x \in \{0,1\}^{b\cdot \ell}$. Notice that for any $x_1,\ldots,x_k \in \{0,1\}^{b\cdot \ell}, 1 \leq t \leq b$, we have
\begin{align*}
    \Big\langle \psi_{b,\ell}(x_1),\ldots,\psi_{b,\ell}(x_k)\Big\rangle &\equiv \Big\langle\CRR_{b,\ell}(x_1,\{q_j\}),\ldots,\CRR_{b,\ell}(x_k,\{q_j\})\Big\rangle \bmod q_t\\
    &\equiv \sum_{d=1}^{\ell}\prod_{i=1}^{k} \CRR(x_{i}^{d},\{q_j\}) \bmod q_t\\
    &\equiv \sum_{d=1}^{\ell} \prod_{i=1}^{k}x_i^{d}[t] \bmod q_t.
\end{align*} Since 
\[
\sum_{d=1}^{\ell} \prod_{i=1}^{k}x_i^{d}[t] \leq \ell
\] and $q_t > \ell$, we know 
\[
\sum_{d=1}^{\ell} \prod_{i=1}^{k}x_i^{d}[t] = 0 \Longleftrightarrow \Big\langle \psi_{b,\ell}(x_1),\ldots,\psi_{b,\ell}(x_k)\Big\rangle  \equiv 0 \bmod q_t.
\] As a result, $\langle x_1,\ldots,x_k\rangle = 0$ if and only if 
\[
\Big\langle \psi_{b,\ell}(x_1),\ldots,\psi_{b,\ell}(x_k)\Big\rangle  \equiv 0 \bmod q_t
\] for all $1 \leq t \leq b$. Finally, $\psi_{b,\ell}(x)_d = \CRR(x^d,\{q_j\}) \leq \ell^{2b}<\ell^{(3k)^{\log^{*}b}\cdot b}$. In addition, 
\[
\Big\langle \psi_{b,\ell}(x_1),\ldots,\psi_{b,\ell}(x_k)\Big\rangle \leq \ell\cdot (\ell^{2b})^k \leq \ell^{(3k)^{\log^{*}b}\cdot kb+1}.
\] Let $V_{b,\ell}$ be the set of all integers in $[0,\ell^{(3k)^{\log^{*}b}\cdot kb+1}]$ that are multiples of all $q_j$'s, and we know
\[
\langle x_1,\ldots,x_k \rangle = 0 \Longleftrightarrow \Big\langle \psi_{b,\ell}(x_1),\ldots,\psi_{b,\ell}(x_k)\Big\rangle \in V_{b,\ell}
\] for all $x_1,\ldots,x_k \in \{0,1\}^{b\cdot \ell}$.

\textbf{Recursive Construction for larger $b$:} When $b \geq \ell$, we let $b_{\micro}$ be such that
\[
\ell^{(3k)^{\log^{*}b_{\micro}}\cdot b_{\micro}} = b
\] and pick $b/b_{\micro}$ primes $q_1,\ldots,q_{b/b_{\micro}}$ in $[b^{k}\ell,b^{2k}\ell^{2}]$. We use induction and assume that the theorem is true for all integers smaller than $b$, and use $q_1,\ldots,q_{b/b_{\micro}}$ for our $\CRR$ encoding. Given a vector $x \in \{0,1\}^{b,\ell}$, we again use $x^i$ to denote the $i$-th block. We further partition $x^i$ into blocks of size $b_{\micro}$, and use $x^{i,j}$ to denote the $j$-th block of $x^i$. We use $x^{[j]}$ to denote the concatenation of the $j$-th micro group in each of the $\ell$ groups. Since $b_{\micro}<b$, we can use induction hypothesis to construct $\psi_{b_{\micro},\ell}:\{0,1\}^{b_{\micro}\cdot \ell} \rightarrow \mathbb{Z}^{\ell}$ such that for every $x,y,z \in \{0,1\}^{b_{\micro}\cdot\ell}$, 
\[
\langle x_1,\ldots,x_k \rangle = 0 \Longleftrightarrow \Big\langle\psi_{b_{\micro},\ell}(x_1),\ldots,\psi_{b_{\micro},\ell}(x_1)\Big\rangle \in V_{b_{\micro},\ell}.
\] We first define an encoding $S: \{0,1\}^{b\cdot \ell} \rightarrow \mathbb{Z}^{b/b_{\micro}\cdot \ell}$:
\begin{align*}
S(x) := \Big(&\psi_{b_{\micro},\ell}(x^{[1]})_1,\psi_{b_{\micro},\ell}(x^{[2]})_1,\ldots,\psi_{b_{\micro},\ell}(x^{[b/b_{\micro}]})_1, \\
& \psi_{b_{\micro},\ell}(x^{[1]})_2,\psi_{b_{\micro},\ell}(x^{[2]})_2,\ldots,\psi_{b_{\micro},\ell}(x^{[b/b_{\micro}]})_2, \\
&\ldots, \ldots, \ldots, \ldots,\\
& \psi_{b_{\micro},\ell}(x^{[1]})_{\ell},\psi_{b_{\micro},\ell}(x^{[2]})_{\ell},\ldots,\psi_{b_{\micro},\ell}(x^{[b/b_{\micro}]})_{\ell}\Bigr).
\end{align*} We partition $S(x)$ into $\ell$ blocks shown above and denote them as $S^{1}(x),\ldots,S^{\ell}(x)$. We then define
\[
\psi_{b,\ell}(x) := \Big(\CRR(S^1(x);\{q_j\}),\CRR(S^2(x);\{q_j\}),\ldots,\CRR(S^{\ell}(x);\{q_j\})\Bigr).
\] Now we show that this encoding satisfies the conditions in the theorem statement. For any $x,y,z \in \{0,1\}^{b\cdot \ell}$, we know
\begin{equation}
\label{eq: construction of psi}
\begin{split}
    \langle x_1,\ldots,x_k \rangle = 0 &\Longleftrightarrow \Big\langle x_1^{[j]},\ldots,x_k^{[j]}\Big\rangle = 0 \quad \forall j \in [b/b_{\micro}]\\
    &\Longleftrightarrow \Big\langle\psi_{b_{\micro},\ell}(x_1^{[j]}),\ldots,\psi_{b_{\micro},\ell}(x_k^{[j]})\Big\rangle \in V_{b_{\micro}} \ \forall j \in [b/b_{\micro}],
\end{split}
\end{equation} where the second equivalence comes from our inductive hypothesis on $\psi_{b_{\micro},\ell}$. By construction of $b_{\micro}$, we have $\psi_{b_{\micro}}(x_i^{[j]})_i < b$ for all $i$, and thus 
\[
\Big\langle\psi_{b_{\micro},\ell}(x_1^{[j]}),\ldots,\psi_{b_{\micro},\ell}(x_k^{[j]}) \Big\rangle<b^{k}\cdot \ell,
\] which implies $V_{\micro} \subseteq [0,b^k\ell)$.

Now for any $1 \leq k \leq b/b_{\micro}$, 
\begin{align*}
    \Big\langle \psi_{b,\ell}(x_1),\ldots,\psi_{b,\ell}(x_k)\Big\rangle &\equiv \sum_{i=1}^{\ell}\big\langle\CRR(S^{i}(x_1);\{q_j\}),\ldots, \CRR(S^{i}(x_k)\Big\rangle\bmod q_t\\
    &\equiv \sum_{i=1}^{\ell} \prod_{d=1}^{k} S^{i}(x_d)_{t}\bmod q_t\\
    &\equiv \sum_{i=1}^{\ell} \prod_{i=1}^{k} \psi_{b_{\micro}}(x_d^{[t]})_i \bmod q_t\\
    &\equiv \Big\langle\psi_{b_{\micro}}(x_1^{[t]}),\ldots,\psi_{b_{\micro}}(x_k^{[t]})\Big\rangle \bmod q_t.
\end{align*} By assumption $q_k \in [b^k\ell,b^{2k}\ell^3]$, so along with \Cref{eq: construction of psi}, we know $\langle x_1,\ldots,x_k \rangle = 0$ is equivalent to 
\[
\Bigl(\Big\langle\psi_{b,\ell}(x_1),\ldots, \psi_{b,\ell}(x_k)\Big\rangle\bmod q_t\Bigr) \in V_{b_{\micro}}
\] for all $1 \leq t \leq b/b_{\micro}$. We finally upper bound $\psi_{b,\ell}(x)_i$ by 
\begin{align*}
    \psi_{b,\ell}(x)_i &< \prod_{j=1}^{b/b_{\micro}}q_j\\
    &\leq (b^{2k}\ell^2)^{b/b_{\micro}}\\
    &\leq 2^{(2k+2)b/b_{\micro}\cdot \log b}\\
    &\leq 2^{(2k+2)b/b_{\micro}\cdot (3k)^{\log^{*}b_{\micro}}\cdot b_{\micro}\cdot \log \ell}\\
    &= \ell^{(2k+2)b\cdot (3k)^{\log^{*}b_{\micro}}}\\
    &\leq \ell^{(3k)^{\log^{*}b}\cdot b},
\end{align*} where the third inequality comes from $b \geq \ell$ and the last inequality comes from 
\[
\log^{*}b_{\micro} \leq \log^{*}\log b = \log^{*}b-1.
\] We set $V_{b}$ to be the set of all integers $c$ in $[0,\ell^{(3k)^{\log^{*}b}\cdot 2b+1}]$ such that $c \bmod q_t \in V_{b_{\micro}}$ for all $1 \leq t \leq b/b_{\micro}$. By our previous analysis $V_b$ will satisfy the statement in the theorem. 

Constructing $\psi_{b,\ell}(x)$ takes $O(\poly(b\cdot\ell))$ time, and constructing $V_b$ takes 
\[
O\Bigl(\ell^{O\big((3k)^{\log^{*}b}\cdot b\big)}\cdot \poly(b\cdot\ell)\Bigr)
\] time because we enumerate all possible values of $\big\langle \psi_{b,\ell}(x_1),\ldots,\psi_{b,\ell}(x_k) \big\rangle$ and for each of them check whether they are in $V_{b_{\micro}}$ after $\bmod q_t$ in $O(\poly(b\cdot \ell))$ time.

\end{proof}

\begin{lemma}
\label{lem: from kOV to kOV of smaller dimension}
Let $1 \leq  \ell \leq d$. There is an $O\Big(n\cdot \ell^{O\big((3k)^{\log^{*}d\cdot d/\ell}\big)}\cdot \poly(d)\Big)$ time reduction from $\kOV_{n,d}$ to $\ell^{O\big((3k)^{\log^{*}d\cdot d/\ell}\big)}$ instances of $\mathbb{Z}\text{-}\threeOV_{n,\ell+1}$, with vectors of entries with bit-length $O(d/\ell\cdot \log \ell \cdot (3k)^{\log^{*}d})$.
\end{lemma}
\begin{proof}
Given $k$ sets $X_1,\ldots,X_k$ of $n$ vectors from $\{0,1\}^d$, we apply $\psi_{d/\ell,\ell}$ to each vector to obtain $X_i' \subset \mathbb{Z}^{\ell}$ for all $1 \leq i \leq k$. By \Cref{thm: function}, 
\[
\exists (x_1,\ldots,x_k) \in X_1 \times \ldots \times X_k \textup{ such that } \langle x_1,\ldots,x_k \rangle = 0\Longleftrightarrow \exists (x_1',\ldots,x_k') \in X_1' \times \ldots \times X_k' \textup{ such that } \langle x_1',\ldots,x_k' \rangle \in V_{d/\ell,\ell}.
\] For each $t \in V_{d/\ell,\ell}$, we construct sets $X_1^t,\ldots,X_k^t$ of vectors from $\mathbb{Z}^{\ell+1}$ such that 
\[
\exists (x_1',\ldots,x_k') \in X_1' \times \ldots \times X_k' \textup{ such that } \langle x_1',\ldots,x_k' \rangle = t \Longleftrightarrow \exists (x_1^t,\ldots,x_k^t)\in X_1^t \times \ldots\times X_k^t \textup{ such that } \langle x_1^t,\ldots,x_k^t \rangle = 0.
\] Indeed, let 
\begin{align*}
    X_1^t &= \{[x_1',-t]: x_1' \in X_1' \}\\
    X_2^t &= \{[x_2',1]: x_2' \in X_2'\}\\
    &\ldots \\
    X_k^t &= \{[x_{k}',1]: x_k' \in X_k'\}.
\end{align*} There are at most $\ell^{O\big((3k)^{\log^{*}d}\cdot d/\ell\big)}$ numbers in $V_{d/\ell,\ell}$, so we have at most $\ell^{O\big((3k)^{\log^{*}d}\cdot d/\ell\big)}$ many $\mathbb{Z}\text{-}\threeOV_{n,\ell+1}$ instances. By \Cref{thm: function}, the reduction takes $O\Big(n\cdot \ell^{O\big((3k)^{\log^{*}d}\cdot d/\ell\big)}\cdot \poly(d)\Big)$ time. The bit-length of reduced vectors is bounded by 
\[
\log \ell^{(3k)^{\log^{*}d/\ell}\cdot d/\ell} = O(\log\ell\cdot d/\ell\cdot (3k)^{\log^{*}d}).
\]
\end{proof}

\begin{theorem}
An oracle solving $\mathbb{Z}\text{-}\kOV_{n,\ell+1}$ where $\ell = (3k+1)^{\log^{*}n}$ in $O(n^{k-\delta})$ time for some $\delta>0$ can be used to construct an $O(n^{k-\delta+o(1)})$ time algorithm for $\kOV_{n,c\log n}$ for an arbitrary constant $c>0$.
\end{theorem}
\begin{proof}
Given $\ell = (3k+1)^{\log^{*}n}$ and any positive constant $c$, we have
\begin{align*}
    \log \ell^{O((3k)^{\log^{*}d}\cdot d/\ell)} &= \log \ell\cdot O((3k)^{\log^{*}d}\cdot d/\ell)\\
    &= O\Big(\frac{\log^{*}n\cdot (3k)^{\log^{*}n}\cdot c\cdot \log n}{(3k+1)^{\log^{*}n}}\Big)\\
    &= O\Big(\log^{*}n\cdot \Big(\frac{3k}{3k+1}\Big)^{\log^{*}n}\cdot c\cdot \log n\Big)\\
    &= o(\log n).
\end{align*}
By the calculation above, \Cref{lem: from kOV to kOV of smaller dimension} shows that we can reduce $\kOV_{n,c\log n}$ to $n^{o(1)}$ instances of $\mathbb{Z}\text{-}\kOV_{n,\ell+1}$ in time $n^{1+o(1)}$. The reduced vectors have bit-length $o(\log n)$. Therefore, we can solve these $n^{o(1)}$ instances with our oracle in time $O(n^{k-\delta+o(1)})$ and this gives an $O(n^{k-\delta+o(1)})$ time algorithm for $\threeOV_{n,d}$.
\end{proof}

The reduction from $\mathbb{Z}\textup{-}\kOV_{n,2^{O(\log^{*}n)}}$ to $\kMinIP_{n,2^{O(\log^{*}n)}}$ is trivial, so we have completed the full reduction from $\SAT$ to $\kMinIP_{n,2^{O(\log^{*}n)}}$, and thereby have the following hardness result.

\begin{theorem}
\label{thm: SAT reduces to kMinIP}
Assuming $\SETH$, there is no algorithm for $\kMinIP_{n,2^{O(\log^{*}n)}}$ with running time $O(n^{k-\varepsilon})$ for any $\varepsilon>0$.
\end{theorem}

\end{document}

%% file: macro.tex
\newtheorem{theorem}{Theorem}[section]
\newtheorem*{theorem*}{Theorem}
\newtheorem{corollary}[theorem]{Corollary}
\newtheorem*{corollary*}{Corollary}

\newtheorem{lemma}[theorem]{Lemma}
\newtheorem{definition}[theorem]{Definition}
\newtheorem*{definition*}{Definition}
\newtheorem{conjecture}[theorem]{Conjecture}

\newcommand{\R}{\mathbb{R}}

\DeclareMathOperator{\poly}{poly}